\documentclass[11pt]{amsart}
\usepackage[english]{babel}
\usepackage{physics}
\usepackage{braket}
\usepackage{dsfont}
\usepackage[table,xcdraw]{xcolor}
\usepackage{multirow}
\usepackage{graphicx}
\usepackage{tikz}
\usepackage{amssymb}
\usepackage{fontawesome}
\usetikzlibrary{matrix}

\usepackage[margin=1in]{geometry}
\usepackage[pagebackref, colorlinks = true, linkcolor = blue, urlcolor  = blue, citecolor = red]{hyperref}

\newcommand{\jewel}{\text{\faDiamond}}
\newcommand{\cuboid}{\text{\faCube}}
\renewcommand{\epsilon}{\varepsilon}

\makeatletter
\newtheorem*{rep@theorem}{\rep@title}
\newcommand{\newreptheorem}[2]{%
\newenvironment{rep#1}[1]{%
 \def\rep@title{#2 \ref{##1}}%
 \begin{rep@theorem}}%
 {\end{rep@theorem}}}
\makeatother

\newtheorem{thm}{Theorem}[section]
\newtheorem*{thm*}{Theorem}
\newreptheorem{thm}{Theorem}
\newtheorem{lem}[thm]{Lemma}
\newtheorem{cor}[thm]{Corollary}

\newtheorem{defi}[thm]{Definition}
\newtheorem{prop}[thm]{Proposition}
\newtheorem{remark}[thm]{Remark}

\begin{document}

\author{Andreas Bluhm}
\email{andreas.bluhm@ma.tum.de}
\address{Zentrum Mathematik, Technische Universit\"at M\"unchen, Boltzmannstrasse 3, 85748 Garching, Germany}
\address{QMATH, Department of Mathematical Sciences, University of Copenhagen, Universitetsparken 5, 2100 Copenhagen, Denmark}

\author{Ion Nechita}
\email{nechita@irsamc.ups-tlse.fr}
\address{Laboratoire de Physique Th\'eorique, Universit\'e de Toulouse, CNRS, UPS, France}
\title[Compatibility and the matrix jewel]{Compatibility of quantum measurements\\and inclusion constants for the matrix jewel}

\begin{abstract}
In this work, we establish the connection between the study of free spectrahedra and the compatibility of quantum measurements with an arbitrary number of outcomes. This generalizes previous results by the authors for measurements with two outcomes. Free spectrahedra arise from matricial relaxations of linear matrix inequalities. A particular free spectrahedron which we define in this work is the matrix jewel. We find that the compatibility of arbitrary measurements corresponds to the inclusion of the matrix jewel into a free spectrahedron defined by the effect operators of the measurements under study. We subsequently use this connection to bound the set of (asymmetric) inclusion constants for the matrix jewel using results from quantum information theory and symmetrization. The latter translate to new lower bounds on the compatibility of quantum measurements. Among the techniques we employ are approximate quantum cloning and mutually unbiased bases. 
\end{abstract}

\date{\today}

\maketitle

\tableofcontents

\section{Introduction}

Given the solution set of a linear matrix inequality, the question often arises whether the unit cube is contained in this set (see Section 1.5 of \cite{helton2019dilations} and references therein). However, this problem, which is known as the matrix cube problem, is known to be NP-hard \cite{ben-tal2002tractable}. Fortunately, there exists a tractable relaxation of this problem which checks inclusion of corresponding free spectrahedra, which are matricial relaxations of the original sets  \cite{ben-tal2002tractable, helton_matricial_2013}. To give error bounds for this relaxation, it is necessary to know the following: if inclusion for the original spectrahedra holds, how much do we have to shrink the smaller free spectrahedron such that inclusion also holds at the level of free spectrahedra? For the matrix cube, as well as for unit balls of $\ell_p$ spaces and other highly symmetric convex sets, these inclusion constants have been recently studied \cite{helton2019dilations,davidson2016dilations,passer2018minimal}. 

Recently, the authors have found that the inclusion constants for the free spectrahedral relaxation of the $\ell_1$-ball, the matrix diamond \cite{davidson2016dilations}, are relevant for the joint measurability of binary quantum measurements \cite{bluhm2018joint}. The fact that not all observables can be measured at the same time is one of the most remarkable properties of quantum mechanics, the observables of position and momentum providing the best-known example of this behavior \cite{Heisenberg1927, Bohr1928}. The notion of joint measurability (or compatibility) has been introduced to capture this property of non-classical theories (see \cite{Heinosaari2016} for a review). In this work, we model quantum measurements by Positive Operator Valued Measures (POVMs), see \cite[Section 3.1]{Heinosaari2011}. POVMs are jointly measurable
 if they arise as marginals from a common measurement. This property is of practical interest, since only POVMs which are not jointly measurable can violate Bell inequalities \cite{Fine1982} or can be used for some quantum information tasks \cite{Brunner2014}.

The present work continues the line of research started in \cite{bluhm2018joint}. While the previous work focused on measurements with only two outcomes, we establish here the connection between the joint measurability of POVMs with an arbitrary number of outcomes and the inclusion of the \emph{matrix jewel}. The matrix jewel is a free spectrahedron which generalizes the matrix diamond and is introduced in this work. We can subsequently use this connection to translate results on joint measurability into bounds on the inclusion constants for the matrix jewel. Some of the techniques used involve approximate cloning of quantum states and mutually unbiased bases. Moreover, we compare the matrix jewel to more symmetric free spectrahedra such as the matrix diamond to obtain lower bounds on the inclusion constants of the matrix jewel. These translate to new bounds on the compatibility of quantum measurements.

We also introduce the notion of \emph{incompatibility witnesses}, which are tuples of self-adjoint matrices that allow, in a simple way, to show that some POVMs are not compatible (the terminology is borrowed from entanglement theory). 

The paper is organized as follows. After presenting informally our main results in Section \ref{sec:main-results}, we recall in Section \ref{sec:preliminaries} some facts from (matricial) convexity theory and quantum information theory; we also introduce at that point a new operation on matrix convex sets, the direct sum. Sections \ref{sec:jewel} and \ref{sec:jewel-compatibility} are the core of the paper: we introduce the matrix jewel and we relate its inclusion properties to compatibility of POVMs. In Section \ref{sec:compatibility-from-QIT}, we use several results from quantum information theory and symmetrization to give lower and upper bounds on the inclusion sets of the matrix jewel. In Sections \ref{sec:incompatibility-witnesses-binary} and \ref{sec:incompatibility-witnesses-general} we develop the theory of incompatibility witnesses. The final section contains a review of our main contributions, as well as some open questions and future research directions.

\section{Main results}\label{sec:main-results}

In this section, we will review the main results of the present work. It is a follow-up paper on the work undertaken in \cite{bluhm2018joint}. We continue investigating the connection between free spectrahedral inclusion problems and joint measurability of quantum effects.

Quantum measurements are identified with \emph{positive operator valued measures} (POVMs). Those are $k$-tuples of positive semidefinite matrices of fixed dimension which sum to the identity. Here, $k$ is the number of measurement outcomes the quantum measurement has. Given a $g$-tuple of POVMs $E^{(1)}, \ldots, E^{(g)}$, where the $i$-th POVM has $k_i$ outcomes, we can ask the question whether these POVMs are \emph{jointly measurable}. Joint measurability means that there is a joint POVM $\Set{G_{j_1, \ldots, j_g}}$ with $j_i \in [k_i]$ from which the POVMs $E^{(i)}$ arise as marginals. Here, we write $[n]$ for the set $\{1, \ldots, n\}$, $n \in \mathbb N$. Although not all measurements in quantum theory are compatible, they can be made compatible if we add a sufficient amount of noise. In this work, we focus on balanced noise, i.e.\ the elements of the $i$-th POVM become
\begin{equation}\label{eq:tilde-Ej}
\tilde{E}_j^{(i)} = s_i E_j^{(i)} + (1-s_i)\frac{1}{k_i}I,
\end{equation}
where $s_i \in [0,1]$. This means, that with probability $s_i$ we measure the original POVM $E^{(i)}$ whereas with probability $1 - s_i$, we output a measurement outcome uniformly at random, independent of the system under study. The set of $g$-tuples $s$ with the property that, for any $g$-tuple of $d$-dimensional POVMs $E^{(i)}$ with $k_i$ outcomes, the noisy POVMs $\tilde E^{(i)}$ from \eqref{eq:tilde-Ej} are compatible, will be written as $\Gamma(g,d,(k_1,\ldots,k_g))$, and will be called the \emph{balanced compatibility region}.

A \emph{free spectrahedron} is a special type of matrix convex set which arises as matricial relaxation of an ordinary linear matrix inequality. The free spectrahedron $\mathcal D_A$ for the self-adjoint matrix $g$-tuple $A$ is the set of self-adjoint matrix $g$-tuples $X$ of arbitrary dimension which fulfill the   matrix inequality
\begin{equation*}
\sum_{i = 1}^g A_i \otimes X_i \leq I.
\end{equation*}
For scalar $X$, we recover the solution set $\mathcal D_A(1)$ of the linear matrix inequality defined by $A$. The free spectrahedral inclusion problem is to determine for which $s \in \mathbb R^{g}_+$ the implication
\begin{equation} \label{eq:implication}
\mathcal D_A(1) \subseteq \mathcal D_B(1) \implies s \cdot \mathcal D_A \subseteq D_B
\end{equation}
is true. Here, we denote by $s \cdot \mathcal D_A$ the set $\{(s_1 X_1, \ldots, s_g X_g): X \in \mathcal D_A\}$. We will be interested in the case where the object on the left hand side is the \emph{matrix jewel}. Consider the free spectrahedron given by the diagonal matrices $\mathrm{diag}[v_j]$, $j \in [k-1]$, where
\begin{equation*}
v_j(\epsilon) = -\frac{2}{k} + 2 \delta_{\epsilon, j} \qquad \forall \epsilon \in [k].
\end{equation*}
We call this spectrahedron the matrix jewel base $\mathcal D_{\jewel,k}$. The matrix jewel $\mathcal D_{\jewel,(k_1, \ldots,k_g)}$ is then the direct sum of the $\mathcal D_{\jewel,k_i}$. We define the direct sum of free spectrahedra arising from polytopes as the maximal spectrahedron which has the direct sum of these polytopes at the scalar level. The precise definition of the matrix jewel can be found in Definition \ref{def:jewel}.
The matrix jewel is a generalization of the matrix diamond introduced in \cite{davidson2016dilations} and considered in relation to quantum effect compatibility in  \cite{bluhm2018joint}. We are interested in the vectors of the form $s = (s_1^{\times (k_1-1)}, \ldots, s_g^{\times (k_g-1)})$ for which the implication in Equation \eqref{eq:implication} is true for $\mathcal D_A = \mathcal D_{\jewel, (k_1,\ldots, k_g)}$ and any self-adjoint tuple $B$ on the right hand side; we are using the notation
$$(s_1^{\times (k_1-1)}, \ldots, s_g^{\times (k_g-1)}) := (\underbrace{s_1, \ldots, s_1}_{k_1-1 \text{ times}}, \ldots, \underbrace{s_g, \ldots, s_g}_{k_g-1 \text{ times}}).$$ 
We call the set of these vectors the \emph{inclusion set for the matrix jewel} $\Delta(g,d,(k_1, \ldots, k_g))$

The main contribution of this work is then the connection of the free spectrahedral inclusion problem to the problem of joint measurability:

\begin{repthm}{thm:jewel-compatible-POVM}
For a fixed matrix dimension $d$, consider $g$ tuples of self-adjoint matrices $E^{(i)} \in (\mathcal M_{d}^{sa})^{k_i - 1}$, $k_i \in \mathbb N$, $i \in [g]$. Define $E^{(i)}_{k_i}:= I_d - E^{(i)}_1 \ldots - E^{(i)}_{k_i -1}$, set $\mathbf k = (k_1, \ldots, k_g)$, and write
\begin{align*}
\mathcal D_{E}&:=\mathcal D_{(2E^{(1)}-\frac{2}{k_1} I,\ldots,  2E^{(g)}-\frac{2}{k_g} I)}\\& = \bigsqcup_{n=1}^\infty \left\{ X \in (\mathcal M_n^{sa})^{\sum_{i=1}^g (k_i-1)} \, : \, \sum_{i=1}^g \sum_{j=1}^{k_i-1} \left(2E^{(i)}_j-\frac{2}{k_i}I\right)\otimes X_{i,j} \leq I_{dn}\right\}.
\end{align*}
Then
\begin{enumerate}
\item $\mathcal D_{\jewel, \mathbf k}(1)  \subseteq \mathcal D_E(1)$ if and only if $\Set{E^{(i)}_1, \ldots, E^{(i)}_{k_i}}$, $i \in [g]$, are POVMs.
\item $\mathcal D_{\jewel, \mathbf k}  \subseteq \mathcal D_E$ if and only if $\Set{E^{(i)}_1, \ldots, E^{(i)}_{k_i}}$, $i \in [g]$, are jointly measurable POVMs.
\item $\mathcal D_{\jewel, \mathbf k}(l)  \subseteq \mathcal D_E(l)$ for $l \in [d]$ if and only if for any isometry $V : \mathbb C^l \hookrightarrow \mathbb C^d$, the tuples $\Set{V^*E^{(i)}_1 V, \ldots, V^* E^{(i)}_{k_i}V}$, $i \in [g]$, are jointly measurable POVMs.
\end{enumerate}
\end{repthm}
This extends \cite[Theorem V.3]{bluhm2018joint} from binary measurements to measurements with $k_i$ outcomes each. We find that the different levels of spectrahedral inclusion correspond to different degrees of joint measurability. 
Furthermore, we show that the balanced compatibility region and the inclusion set for the matrix jewel can be identified; again, this is a generalization of \cite[Theorem V.7]{bluhm2018joint} for an arbitrary number of outcomes.
\begin{repthm}{thm:delta-is-gamma}
Let $d$, $g \in \mathbb N$ and $(k_1, \ldots, k_g) \in \mathbb N^{g}$. Then,
\begin{equation*}
\Gamma(g,d,(k_1, \ldots, k_g)) = \Delta(g,d,(k_1, \ldots, k_g)).
\end{equation*}
\end{repthm}

This identification allows to use results on one set to characterize the other. In \cite{bluhm2018joint}, we mostly adapted results from the study of free spectrahedral inclusion to characterize the balanced compatibility region in quantum information theory. This was possible, since the matrix diamond (the matrix jewel for $k_i = 2$ for all $i \in \mathbb N$) is a highly symmetric object and has already been studied in the literature. The matrix jewel does not have these symmetries and has not been studied in the algebraic convexity literature. Therefore, we adapt results from quantum information theory in Section \ref{sec:compatibility-from-QIT}, which we subsequently use in Section \ref{sec:discussion} to give upper and lower bounds on $\Delta(g,d,(k_1, \ldots, k_g))$. The lower bounds come from asymmetric approximate cloning of quantum states and from two different symmetrization procedures. The latter yield new lower bounds on the balanced compatibility region of quantum measurements. General upper bounds can be imported from the case of binary POVMs, since more outcomes shrink the compatibility regions and therefore also the corresponding inclusion sets. For the case of $k_i = d$ and $g$ not too large, we get better bounds from the study of measurements arising from mutually unbiased bases (MUBs).

We also introduce in this paper the notion of \emph{incompatibility witnesses}, both in the case of binary POVMs (Section \ref{sec:incompatibility-witnesses-binary}) and general POVMs (Section \ref{sec:incompatibility-witnesses-general}). As in the case of compatibility conditions, the theory in the binary case is simpler and the corresponding free spectrahedra have already been studied extensively in the mathematical literature. For these reasons, let us focus here on binary POVMs. 

A $g$-tuple of self-adjoint matrices $X \in (\mathcal M_n^{sa})^g$ is called an incompatibility witness if $X$ is an element of the matrix diamond $\mathcal D_{\diamondsuit,g}$, i.e.\ if $\sum_{i=1}^g \epsilon_i X_i \leq I_n$ for all sign vectors $\epsilon \in \{\pm 1\}^g$. An incompatibility witness $X$ can certify that $g$ given effects are incompatible: If the  matrix inequality
$$\sum_{i=1}^g (2E_i - I_d) \otimes X_i \leq I_{dn}$$
does not hold, the effects $E_1, \ldots, E_g$ are incompatible. There is a strong connection between incompatibility witnesses and the matrix cube (arguably the most studied class of free spectrahedra): $X$ is an incompatibility witness if and only if $\mathcal D_{\square, g}(1) \subseteq \mathcal D_{X}(1)$. Using the inclusion constants for the (complex) matrix cube, one can obtain tractable relaxations for the two equivalent conditions above (which otherwise require checking an exponential number of matrix inequalities).

\section{Preliminaries}\label{sec:preliminaries}

This section contains some facts from (algebraic) convexity and quantum information theory which will be needed in the following sections. The material here is for the most part well known, with the exception of Section \ref{sec:direct-sum}.

\subsection{Convex analysis}
Before we move on to the main topic of this section, let us fix some basic notation. We will often write $[n] := \Set{1, \ldots, n}$ for brevity, where $n \in \mathbb N$. Furthermore, we will use $\mathbb R^g_+:= \Set{x \in \mathbb R^g: x_i \geq 0~\forall i \in [g]}$, where $g \in \mathbb N$. Let $n$, $m \in \mathbb N$. Then, $\mathcal M_{n,m}$ is the set of complex $n \times m$ matrices and we will write just $\mathcal M_n$ if $m = n$. For the self-adjoint matrices, we will write $\mathcal M_n^{sa}$. By $\mathcal U(d)$ we will denote the unitary $d \times d$ matrices. Moreover, we will write $I_n \in \mathcal M_n$ for the identity matrix, where we will often omit the subscript if the dimension is clear from the context. The \emph{operator system} generated by the $g$-tuple $A \in (\mathcal M_d^{sa})^g$ is defined as
\begin{equation*}
\mathcal{OS}_A:= \mathrm{span}\Set{I_d,A_i: i \in [g]}.
\end{equation*}
Here, the span is taken over the complex numbers. Furthermore, we will often write for such $g$-tuples $2A -I := (2A_1 - I_d, \ldots, 2A_g - I_d)$ and $V^\ast A V:= (V^\ast A_1 V, \ldots, V^\ast A_g V)$ with $V \in \mathcal M_{d,k}$, $k \in \mathbb N$.

We start with two standard objects in convex analysis, polytopes and polyhedra (c.f.\ \cite[Definition I.2.2]{Barvinok2002}).
\begin{defi}\label{def:poly}
The convex hull of a finite set of points in $\mathbb R^d$, $d \in \mathbb N$, is called a \emph{polytope}. Let $c_1, \ldots, c_m$ be vectors in $\mathbb R^d$ and let $\alpha_1, \ldots, \alpha_m \in \mathbb R$. The set 
\begin{equation*}
\mathcal P := \Set{x \in \mathbb R^d: \langle c_i, x \rangle \leq \alpha_i \quad \forall i \in [m]}
\end{equation*}
is called a \emph{polyhedron}.
\end{defi}
By the Weyl-Minkowski theorem, a convex subset of $\mathbb R^d$ is a polytope if and only if it is a bounded polyhedron \cite[Corollary II.4.3]{Barvinok2002}. We will need the following lemma, which follows easily from convexity:
\begin{lem}[{\cite[Section IV.1]{Barvinok2002}}] \label{lem:polarextremepoints}
Let $\mathcal P = \mathrm{conv}(\Set{v_1, \ldots v_m}) \subset \mathbb R^d$, $m \in \mathbb N$. Then, its polar dual can be written as $\mathcal P^\circ = \Set{x \in \mathbb R^d: \langle v_i, x\rangle \leq 1 \quad \forall i \in [m]}$.
\end{lem}
There are several ways of constructing new convex sets from a collection of given ones. One way is the Cartesian product:
\begin{defi}
Let $\mathcal P_1 \subseteq \mathbb R^{k_1}$, $\mathcal P_2 \subseteq \mathbb R^{k_2}$ be two convex sets. Then, their \emph{Cartesian product} is
\begin{equation*}
\mathcal P_1 \times \mathcal P_2 := \Set{(x,y) \in \mathbb R^{k_1 + k_2}: x \in \mathcal P_1, y \in \mathcal P_2}.
\end{equation*} 
\end{defi}
Another one is the direct sum:
\begin{defi} \label{def:direct-sum-convex-sets}
Let $\mathcal P_1 \subseteq \mathbb R^{k_1}$, $\mathcal P_2 \subseteq \mathbb R^{k_2}$ be two convex sets. Then, their \emph{direct sum} is
\begin{equation*}
\mathcal P_1 \oplus \mathcal P_2 := \mathrm{conv}\left(\Set{(x,0) \in \mathbb R^{k_1 + k_2}: x \in \mathcal P_1} \cup \Set{(0,y) \in \mathbb R^{k_1 + k_2}: y \in \mathcal P_2}\right).
\end{equation*} 
\end{defi}
\begin{remark}
In particular, the above definition shows that the direct sum of two polytopes is again a polytope, because it is the convex hull of their respective extreme points embedded into a higher dimensional space.
\end{remark}
We can find a useful expression for the direct sum of two polytopes in terms of the Cartesian product and taking polars. We include a short proof for convenience.
\begin{lem}[{\cite[Lemma 2.4]{Bremner1997}}] \label{lem:bremner}
Let $\mathcal P_1 \subseteq \mathbb R^{k_1}$, $\mathcal P_2 \subseteq \mathbb R^{k_2}$ be two polytopes and such that $0 \in \mathcal P_1$, $\mathcal P_2$. Then, 
\begin{equation*}
\mathcal P_1 \oplus \mathcal P_2 = (\mathcal P_1^\circ \times \mathcal P_2^\circ)^\circ.
\end{equation*}
\end{lem}
\begin{proof}
Using Lemma \ref{lem:polarextremepoints}, we may write
\begin{equation*}
(\mathcal P_1 \oplus \mathcal P_2)^\circ = \Set{(x_1,x_2) \in \mathbb R^{k_1} \times  \mathbb R^{k_2} : \langle p_i, x_i \rangle \leq 1 \quad \forall p_i \in \mathcal P_i, ~i \in [2]}.
\end{equation*}
Comparing this with the definition of $\mathcal P_i^\circ$, we find that $(\mathcal P_1 \oplus \mathcal P_2)^\circ = \mathcal P_1^\circ \times \mathcal P_2^\circ$. As the $\mathcal P_i$ are polytopes, they are compact and thus also $\mathcal P_1 \oplus \mathcal P_2$ is compact. As this set furthermore contains $0$ by assumption, an application of the Bipolar Theorem \cite[Theorem IV.1.2]{Barvinok2002} yields the claim.
\end{proof}

Later, we shall need the following result on the faces of the Cartesian product.
\begin{lem}[{\cite[Lemma 2.3]{Bremner1997}}] \label{lem:cartesianfaces}
Let $\mathcal P_1 \subseteq \mathbb R^{k_1}$, $\mathcal P_2 \subseteq \mathbb R^{k_2}$ be two polytopes. Then, the $l$-dimensional faces of $\mathcal P_1 \times \mathcal P_2$, for $0 \leq l \leq k_1 + k_2 $ are are of the form $\mathcal F_1 \times \mathcal F_2$, where $\mathcal F_i$ is a $j_i$-dimensional face of $\mathcal P_i$ and $j_1 + j_2 = l$.
\end{lem}

\subsection{Matrix convex sets and free spectrahedra}

In this section, we will review some basic results from the theory of matrix convex sets and free spectrahedra. The theory we will need for this work can be found in \cite{helton_matricial_2013, helton2019dilations, davidson2016dilations}. We shall write $\mathrm{UCP}(\mathcal B(\mathcal H), \mathcal B(\mathcal K))$ for the set of unital completely positive maps from the bounded operators on a Hilbert space $\mathcal H$ to bounded operators on a Hilbert space $\mathcal K$.

\begin{defi}
Let $g \in \mathbb N$. Moreover, let $\mathcal F_n \subseteq (\mathcal M_n^{sa})^g$ for all $n \in \mathbb N$. Then, we call $\mathcal F = \bigsqcup_{n \in \mathbb N} \mathcal F_n$ a \emph{free set}. Moreover, $\mathcal F$ is a \emph{matrix convex set} if it satisfies the following two properties for any $m$, $n \in \mathbb N$:
\begin{enumerate}
\item If $X = (X_1, \ldots, X_g) \in \mathcal F_m$, $X = (Y_1, \ldots, Y_g)\in \mathcal F_n$, then $X \oplus Y := (X_1 \oplus Y_1, \ldots, X_g \oplus Y_g) \in \mathcal F_{m+n}$
\item If $X = (X_1, \ldots, X_g) \in \mathcal F_m$ and $\Psi: \mathcal M_m \to \mathcal M_n$ is a unital completely positive (UCP) map, then $(\Psi(X_1), \ldots, \Psi(X_g)) \in \mathcal F_n$.
\end{enumerate}
That is, a matrix convex set is a free set closed under direct sums and UCP maps.
\end{defi}

This definition can seen to be equivalent to the one used in \cite{Wittstock1984} (see \cite[Section 2]{davidson2016dilations}). In particular, it follows from the definition that all sets $\mathcal F_n$ are convex. A matrix convex set $\mathcal F$ is open/closed/bounded if all $\mathcal F_n$ defining it have this property. There are two important examples of classes of matrix convex sets. The first are free spectrahedra:
\begin{defi}
Let $\mathcal H$ be a Hilbert space and $A \in (\mathcal B(\mathcal H)^{sa})^g$ be a $g$-tuple of self-adjoint bounded operators on this Hilbert space. The \emph{free spectrahedron at level $n$} defined by $A$ is the set
\begin{equation} \label{eq:lmi}
\mathcal D_A(n) := \Set{X \in (\mathcal M_n^{sa})^g: \sum_{i = 1}^g A_i \otimes X_i \leq I_{\mathcal H} \otimes I_n}.
\end{equation}
Here, $I_{\mathcal H}$ is the identity operator on $\mathcal H$. The \emph{free spectrahedron} corresponding to $A$ is then the union of all these levels, i.e.\
\begin{equation*}
\mathcal D_A:= \bigsqcup_{n \in \mathbb N} \mathcal D_A(n).
\end{equation*}
\end{defi}
It is easy to see that free spectrahedra are closed matrix convex sets \cite[Propostion 2.1]{davidson2016dilations}. Some authors consider free spectrahedra defined by tuples of matrices (i.e.~$\mathcal H$ is finite dimensional) \cite{helton_matricial_2013, helton2019dilations}, but the notion can be extended to bounded operators on arbitrary Hilbert spaces \cite{davidson2016dilations}. Most of this work will only consider the case where $A$ is a tuple of self-adjoint matrices.

The second class of examples comes from the matrix ranges introduced in \cite{Arveson1972} and generalized in \cite{davidson2016dilations}:
\begin{defi} Let $\mathcal H$ be a Hilbert space and $g \in \mathbb N$. Then, the \emph{matrix range} $\mathcal W(A)$ of $A = (A_1, \ldots, A_g) \in (\mathcal B(\mathcal H)^{sa})^g$ is defined as $\mathcal W = \bigsqcup_{n \in \mathbb N} \mathcal W_n$, where for any $n \in \mathcal N$
\begin{equation*}
\mathcal W_n(A):= \left\{(X_1, \ldots, X_g) \in (\mathcal M_n^{sa})^g: \exists \Psi \in \mathrm{UCP}(\mathcal B(\mathcal H),\mathcal M_n)\mathrm{~s.t.~}X_i = \Psi(A_i)~\forall i \in [g]\right\}.
\end{equation*}
\end{defi}
It is again easy to see that matrix ranges are closed bounded matrix convex sets \cite[Propositon 2.5]{davidson2016dilations}. Let us point out that the two examples of matrix convex sets discussed above are paradigmatic \cite[Proposition 3.5]{davidson2016dilations}: a closed matrix convex set is bounded if and only if it is a matrix range and it contains $0$ in its interior if and only if it is a free spectrahedron. 

As for usual convex sets, we can define the polar dual of a matrix convex set. We extend the definition given in \cite{Effros1997} to free sets instead of restricting to matrix convex sets.
\begin{defi}\label{def:mcs-polar}
Let $g \in \mathbb N$ and let $\mathcal F$ be a free set $\mathcal F = \bigsqcup_{n \in \mathbb N} \mathcal F_n$, $n \in \mathbb N$. Then, its \emph{polar dual} is defined as $\mathcal F^{\bullet} = \bigsqcup_{n \in \mathbb N} \mathcal F^{\bullet}_n$, where 
\begin{equation*}
\mathcal F^{\bullet}_n := \left\{ X \in (\mathcal M_n^{sa})^g: \sum_{i = 1}^g A_i \otimes X_i \leq I ~\forall A \in \mathcal F \right\}.
\end{equation*}
\end{defi}
It is easy to verify that $\mathcal F^{\bullet}$ is a closed matrix convex set containing $0$. It has been shown in \cite{davidson2016dilations} that matrix ranges and free spectrahedra are polar duals of each other.
\begin{prop}[{\cite[Proposition 3.1 and 3.3]{davidson2016dilations}}] \label{prop:polar-duals-of-FS}
Let $\mathcal H$ be a Hilbert space and  $A \in (\mathcal B(\mathcal H)^{sa})^g$, $g \in \mathbb N$. Then, $\mathcal W(A)^\bullet = \mathcal D_A$. Moreover, if $0 \in \mathcal W(A)$, then $(\mathcal D_{A})^\bullet = \mathcal W(A)$.
\end{prop}
Lemma 3.4 of \cite{davidson2016dilations} shows that  $0 \in \mathcal W(A)$ is equivalent to $\mathcal D_A(1)$ being bounded. 

Let $\mathcal C \subseteq \mathbb R^g$ be a convex set. In general, there are many free spectrahedra $\mathcal D_A$ with $\mathcal D_A(1) = \mathcal C$. If $\mathcal C$ is a polyhedron with 0 in its interior, we can find a maximal such free spectrahedron \cite[Definition 4.1]{davidson2016dilations}:
\begin{align}
\label{eq:Wmax}&\mathcal W_{max}(\mathcal C)(n) :=\\
\nonumber& \Set{X \in (\mathcal M_n^{sa})^g: \sum_{i = 1}^g c_i X_i \leq \alpha I, \quad \forall \, c \in \mathbb R^g, \forall \alpha \in \mathbb R~\mathrm{\ s.t.\ }~\mathcal C \subseteq \Set{x \in \mathbb R^g: \langle c, x\rangle \leq \alpha}}.
\end{align}
Note that $\mathcal W_{max}(\mathcal C)(1) = \mathcal C$, as claimed above.
\begin{remark}
It is clear that the above is indeed a free spectrahedron defined by matrices for a polyhedron $\mathcal C$, since polyhedra are defined as the intersection of \emph{finitely} many hyperplanes (see Definition \ref{def:poly}). The defining matrices can thus be chosen diagonal and of finite dimension. The fact that $0$ is an interior point guarantees that we can always choose $\alpha = 1$.
\end{remark}
\begin{remark}
The definition above can be used to define the largest matrix convex set $\mathcal F$ with $\mathcal F_1 = \mathcal C$ for any convex set $\mathcal C$. If $\mathcal C$ is not a polyhedron or $0$ not in the interior, however, the corresponding $\mathcal W_{max}(\mathcal C)$ is not necessarily a free spectrahedron defined by matrices. See \cite[Section 4]{davidson2016dilations} for details.
\end{remark}

In this work, we will be concerned with inclusion constants, i.e.\ constants for which the implication 
\begin{equation*}
 \mathcal  D_A(1) \subseteq \mathcal D_B(1) \implies s\cdot \mathcal D_A \subseteq \mathcal D_B
\end{equation*}
holds, where $A$, $B$ are both $g$-tuples of self-adjoint matrices. Here, the \emph{(asymmetrically) scaled} free spectrahedron is 
\begin{equation*}
s \cdot \mathcal D_A:= \Set{(s_1 X_1, \ldots, s_g X_g): X \in \mathcal D_A}.
\end{equation*}
\begin{defi}
Let $D \in \mathbb N$ and $\mathcal D_A$ be the free spectrahedron defined by $ A:= (A^{(1)},\ldots,  A^{(g)})$, where $A^{(j)} \in (\mathcal M_D^{sa})^{k_j-1}$, $k_j \in \mathbb N$, $j \in [g]$. Let $\mathbf k = (k_1, \ldots, k_g)$. The \emph{inclusion set} is defined as 
\begin{align*}
&\Delta_{\mathcal D_A}(g,d, \mathbf k):=\\ &\Set{s \in \mathbb R^g_+: \forall B \in (\mathcal M_d^{sa})^{\sum_{i = 1}^g (k_i -1)}, \,  \mathcal D_A(1) \subseteq \mathcal D_B(1) \implies (s_1^{\times (k_1-1)}, \ldots, s_g^{\times (k_g-1)}) \cdot \mathcal D_A \subseteq \mathcal D_B}.
\end{align*}
If $\mathcal D_A$ is the matrix jewel $\mathcal D_{\jewel, \mathbf k}$ in Definition \ref{def:jewel}, we will write $\Delta$ instead of $\Delta_{\mathcal D_A}$.
\end{defi}
This definition generalizes \cite[Definition IV.1]{bluhm2018joint}, which is recovered for $\mathbf k = (2, \ldots, 2)$. Note that the $(k_i -1)$-tuples in the inclusion sets are scaled in the same way inside each group, where the size of the groups are determined by the vector $\mathbf k$. By the same argument as in \cite[Proposition IV.3]{bluhm2018joint}, these sets are convex.

The inclusion of free spectrahedra can be related to positivity properties of the map between the matrices defining them. Let $A \in (\mathcal M_D^{sa})^g$ be such that the $A_i$, $i \in [g]$, are linearly independent and let $B \in (\mathcal M_d^{sa})^g$. Let $\Phi^{A \to B}: \mathcal{OS}_A \to \mathcal M_d$ be the unital map defined by
\begin{equation*}
\Phi^{A \to B}: A_i \mapsto B_i \qquad \forall i \in [g].
\end{equation*}
If there is no confusion, we will later drop the superscript for convenience. Note that the assumption on the linear independence of the $A_i$, $i \in [g]$, is in particular met if $D_A(1)$ is bounded. The following theorem has been proven in \cite[Theorem 3.5]{helton_matricial_2013} for real matrices. See \cite[Lemma IV.4]{bluhm2018joint} for a very similar proof in the complex case.
\begin{lem}\label{lem:phi}
Let $A \in (\mathcal M_D^{sa})^g$ and $B \in (\mathcal M_d^{sa})^g$. Furthermore, let $\mathcal D_A(1)$ be bounded. Then, $\mathcal D_A(n) \subseteq \mathcal D_B(n)$ holds if and only if $\Phi^{A \to B}$ as given above is $n$-positive. In particular, $\mathcal D_A \subseteq \mathcal D_B$ if and only if $\Phi^{A \to B}$ is completely positive.
\end{lem}

\subsection{The direct sum of matrix convex sets}\label{sec:direct-sum}
In this section, we introduce the \emph{direct sum} of matrix convex sets and compare it to other existing operations on matrix convex sets. Subsequently, we relate it to the direct sum of polytopes. We derive some simple properties of this construction which will be used later in the paper. Here, we will identify $\mathbb R^d$ with the diagonal $d \times d$ matrices with real entries. 

In order to construct new matrix convex sets, we can define a Cartesian product on them:
\begin{defi}
Let $\mathcal F$, $\mathcal G$ be two free sets. Their \emph{Cartesian product} is defined as $\mathcal F \hat\times \mathcal G := \bigsqcup_{n \in \mathbb N} (\mathcal F \hat\times \mathcal G)_n$, where
\begin{equation*}
(\mathcal F \hat\times \mathcal G)_n := \left\{ (X,Y): X \in \mathcal F_n, Y \in \mathcal G_n\right\}.
\end{equation*}
\end{defi}
It is easy to check that for matrix convex sets $\mathcal F$ and $\mathcal G$, the set $\mathcal F \hat \times \mathcal G$ is also matrix convex. Moreover, the Cartesian product of matrix convex sets at level $n=1$ is the ordinary Cartesian product of convex sets, i.e.\
$$(\mathcal F \hat\times \mathcal G)_1 = \mathcal F_1 \times \mathcal G_1.$$ The Cartesian product of matrix convex sets has been used previously, see e.g.\ \cite{Passer2018a}. We use the same definition as the recent paper \cite{Passer2019} (see Definition 4.1 in said paper), but allow for arbitrary free sets. In the case where both $\mathcal F$ and $\mathcal G$ are free spectrahedra, their Cartesian product is again a free spectrahedron for which we can give an explicit form: 
\begin{prop}
Let $A \in (\mathcal B(\mathcal H_1)^{sa})^{k_1}$, $B \in (\mathcal B(\mathcal H_2)^{sa})^{k_2}$ be tuples of self-adjoint bounded operators, where $k_1$, $k_2 \in \mathbb N$. Then $\mathcal D_A \hat\times \mathcal D_B$ is the free spectrahedron defined as
\begin{align}
&(\mathcal D_A \hat\times \mathcal D_B)(n) = \label{eq:def-Cartesian-product-free-spectrahedra} \\&\Set{X \in (\mathcal M_n^{sa})^{k_1 + k_2}: \sum_{i = 1}^{k_1} (A_i \oplus 0_{\mathcal H_2}) \otimes X_i + \sum_{j = 1}^{k_2} (0_{\mathcal H_1} \oplus B_j) \otimes X_{k_1 + j} \leq I_{\mathcal H_1\oplus \mathcal H_2} \otimes I_{n}}. \nonumber
\end{align} 
\end{prop}
\begin{proof}
The assertion follows since $(\mathcal H_1 \oplus \mathcal H_2) \otimes \mathbb C^n \simeq (\mathcal H_1 \otimes \mathbb C^n) \oplus (\mathcal H_2 \otimes \mathbb C^n)$. Thus, 
\begin{equation*}
 \sum_{i = 1}^{k_1} (A_i \oplus 0_{\mathcal H_2}) \otimes X_i + \sum_{j = 1}^{k_2} (0_{\mathcal H_1} \oplus B_j) \otimes X_{k_1 + j} \leq I_{\mathcal H_1\oplus \mathcal H_2} \otimes I_{n}
\end{equation*}
if and only if
\begin{equation*}
\left( \sum_{i = 1}^{k_1} A_i \otimes X_i \right) \oplus \left(\sum_{j = 1}^{k_2} B_j \otimes X_{k_1 + j} \right) \leq (I_{\mathcal H_1} \otimes I_n) \oplus (I_{\mathcal H_2} \otimes I_n).
\end{equation*}
The above holds if and only if both $(X_1, \ldots, X_{k_1}) \in \mathcal D_A(n)$ and $(X_{k_1+1}, \ldots, X_{k_1+k_2}) \in \mathcal D_B(n)$.
\end{proof}

In the same way as $\mathcal F \hat \times \mathcal G$ generalizes the Cartesian product of convex sets, we now define a direct sum of matrix convex sets which generalizes the direct sum of convex sets, using the duality notion introduced in Definition \ref{def:mcs-polar}.

\begin{defi} \label{defi:mc_direct_sum}
Let $\mathcal F$, $\mathcal G$ be two matrix convex sets defined by $\mathcal F_n \in (\mathcal M_n^{sa})^{g_1}$ and $\mathcal G_n \in (\mathcal M_n^{sa})^{g_2}$ for all $n \in \mathbb N$, respectively. Their direct sum is defined as
\begin{equation*}
(\mathcal F \hat \oplus \mathcal G) := ( (\mathcal F^\bullet \otimes I) \hat \times (I \otimes \mathcal G^{\bullet}))^\bullet
\end{equation*}
Here, $(\mathcal F^\bullet \otimes I)_{n^2} := \{(X_1 \otimes I_n, \ldots, X_{g_1} \otimes I_n): X \in F^\bullet_n \}$ for all $n \in \mathbb N$ and $(\mathcal F^\bullet \otimes I)_{m} = \emptyset$ for all other $m \in \mathbb N$. The free set $I \otimes \mathcal G^{\bullet}$ is defined analogously.
\end{defi}
Since $(\mathcal F \hat \oplus \mathcal G)$ is the polar of a free set, it is a closed matrix convex set containing $0$. Again, we find that for $\mathcal F$ and $\mathcal G$ two free spectrahedra, their direct sum is a free spectrahedron as well and we can give an explicit description of it. Before we can do that, we need to prove a lemma.

\begin{lem} \label{lem:positive-on-tensor-subspaces}
Let $\mathcal H_1$, $\mathcal H_2$ be two Hilbert spaces, $A \in \mathcal B(\mathcal H_1 \otimes \mathcal H_2 \otimes \mathbb C^d)^{sa}$, where $d \in \mathbb N$. If for all $n \in \mathbb N$
\begin{equation*}
(P \otimes Q \otimes I_d) A (P \otimes Q \otimes I_d) \geq 0 
\end{equation*}
for all  $P, Q$ orthogonal projections onto $n$-dimensional subspaces of $\mathcal H_1$ and $\mathcal H_2$, then $A \geq 0$.
\end{lem}
\begin{proof}
Let $\{e_\alpha\}_{\alpha \in A_1}$, $\{f_\beta\}_{\beta \in A_2}$ be orthonormal bases of $\mathcal H_1$ and $\mathcal H_2$, respectively. Here, $A_1$, $A_2$ are not necessarily countable sets. Moreover, let $\{g_k\}_{k \in [d]}$ be an orthonormal basis of $\mathbb C^d$. Then, $\{e_\alpha \otimes f_\beta \otimes g_k: \alpha \in A_1, \beta \in A_2, k \in [d]\}$ is an orthonormal basis of $H_1 \otimes \mathcal H_2 \otimes \mathbb C^d$ \cite[Proposition 2 in Section II.4]{Reed1980}. Let us assume that $A$ is not positive. Then, there exists a $\psi \in \mathcal H_1 \otimes \mathcal H_2 \otimes \mathbb C^d$ with $\norm{\psi} =1$ such that $\langle \psi, A \psi \rangle < 0$. We can write $\psi$ in a basis as
\begin{equation*}
\psi = \sum_{i,j = 1}^\infty \sum_{k = 1}^d \psi_{ijk} e_i \otimes f_j \otimes g_k
\end{equation*}
where $\psi_{ijk} \in \mathbb C$ for all $i$, $j \in \mathbb N$, $k \in [d]$ and the series converges in norm \cite[Theorem II.6]{Reed1980}. Let us define 
\begin{equation*}
\psi_N = \sum_{i,j = 1}^N \sum_{k = 1}^d \psi_{ijk} e_i \otimes f_j \otimes g_k
\end{equation*}
for $N \in \mathbb N$. Since the series for $\psi$ converges in norm, for every $\epsilon > 0$ $\exists N \in \mathbb N$ such that $\norm{\psi - \psi_N} \leq \epsilon$ and thus
\begin{equation*}
|\langle \psi_N, A \psi_N \rangle - \langle \psi, A \psi \rangle| \leq 2\epsilon \norm{A}_\infty  
\end{equation*}
can be seen from Bessel's inequality \cite[Corollary to Theorem II.1]{Reed1980} and the Cauchy-Schwarz-inequality. Therefore, we find that for $N$ large enough, $\langle \psi_N, A \psi_N\rangle <0$. We choose $P$ and $Q$ to be the orthogonal projections onto the space spanned by $\{e_i\}_{i \in [N]}$ and $\{f_j\}_{j \in [N]}$, respectively. Then, 
\begin{equation*}
\langle \psi_N, (P \otimes Q \otimes I_d) A (P \otimes Q \otimes I_d) \psi_N \rangle  = \langle \psi_N, A \psi_N \rangle < 0,
\end{equation*}
which contradicts the assumption.
\end{proof}

\begin{prop} \label{prop:direct-sum-FS}
Let $A \in (\mathcal B(\mathcal H_1)^{sa})^{k_1}$ and $B \in (\mathcal B(\mathcal H_2)^{sa})^{k_2}$, where $k_1$, $k_2 \in \mathbb N$. Moreover, let us assume that $\mathcal D_A(1)$ and $\mathcal D_B(1)$ are bounded. Then, their direct sum $\mathcal D_A \hat\oplus \mathcal D_B$ is the free spectrahedron defined as
\begin{align} 
&\mathcal D_A \hat\oplus \mathcal D_B (n) =\label{eq:direct-sum-FS}  \\ &\Set{X \in (\mathcal M_n^{sa})^{k_1 + k_2}: \sum_{i = 1}^{k_1} (A_i \otimes I_{\mathcal H_2}) \otimes X_i + \sum_{j = 1}^{k_2} (I_{\mathcal H_1} \otimes B_j) \otimes X_{k_1 + j} \leq I_{\mathcal H_1\otimes\mathcal H_2}\otimes I_n}.\nonumber
\end{align}
\end{prop}
\begin{proof}
We note that the set on the right hand side of Equation \eqref{eq:direct-sum-FS} is $\mathcal D_{(A \otimes I_{\mathcal H_2}, I_{\mathcal H_1} \otimes B)}$. The boundedness of $\mathcal D_A(1)$ and $\mathcal D_B(1)$ implies by \cite[Lemma 3.4]{davidson2016dilations} that $0 \in \mathcal W(A)$ and $0 \in \mathcal W(B)$. From Proposition \ref{prop:polar-duals-of-FS}, we infer that 
\begin{equation*}
\mathcal D_A^\bullet = \mathcal W(A), \qquad \mathcal D_B^\bullet = \mathcal W(B).
\end{equation*}
Thus, 
\begin{align*}
((\mathcal D_A^\bullet \otimes I) \hat \times (I \otimes \mathcal D_B^\bullet))_{n^2} =& \{ (X \otimes I_n, I_n \otimes Y):~\forall i \in [2]~\exists \Psi_i \in \mathrm{UCP}(\mathcal B(\mathcal H_i), \mathcal M_n)\mathrm{~s.t.~}\Psi := \Psi_1 \otimes \Psi_2\\& (X \otimes I_n, I_n \otimes Y) = (\Psi(A_1 \otimes I_{\mathcal H_2}),\Psi(A_2 \otimes I_{\mathcal H_2}), \ldots, \Psi(I_{\mathcal H_1} \otimes B_{k_2})) \}.
\end{align*}
and all other $((\mathcal D_A^\bullet \otimes I) \hat \times (I \otimes \mathcal D_B^\bullet))_{m}$, $m \in \mathbb N$, are empty. Hence, we have the inclusion
\begin{equation*}
(\mathcal D_A^\bullet \otimes I) \hat \times (I \otimes \mathcal D_B^\bullet) \subseteq \mathcal W(A \otimes I_{\mathcal H_2}, I_{\mathcal H_1} \otimes B),
\end{equation*}
which implies 
\begin{equation*}
\mathcal D_A \hat\oplus \mathcal D_B \supseteq \mathcal D_{(A \otimes I_{\mathcal H_2}, I_{\mathcal H_1} \otimes B)}
\end{equation*}
by Proposition \ref{prop:polar-duals-of-FS}. For the reverse inclusion, let $X \in (\mathcal D_A \hat\oplus \mathcal D_B)_n$, $n \in \mathbb N$. Then, for all UCP maps $\Psi_i: \mathcal B(\mathcal H_i) \to \mathcal M_m$, $m \in \mathbb N$, $i \in [2]$,
\begin{equation*}
\sum_{i = 1}^{k_1} [(\Psi_1\otimes \Psi_2)(A_i \otimes I_{\mathcal H_2})] \otimes X_i + \sum_{j = 1}^{k_2} [(\Psi_1\otimes \Psi_2)(I_{\mathcal H_1} \otimes B_j)] \otimes X_{k_1 + j} \leq I_{m^2n}.
\end{equation*}
Consider now orthogonal projections $P, Q$ onto $m$-dimensional subspaces of $\mathcal H_{1,2}$, respectively. In particular $\Psi_i:X \mapsto PXP$ and $\Psi_i:Y \mapsto QYQ$ are valid UCP maps. Then, Lemma \ref{lem:positive-on-tensor-subspaces} implies that 
\begin{equation*}
\sum_{i = 1}^{k_1} (A_i \otimes I_{\mathcal H_2}) \otimes X_i + \sum_{j = 1}^{k_2} (I_{\mathcal H_1} \otimes B_j) \otimes X_{k_1 + j} \leq I_{\mathcal H_1\otimes\mathcal H_2}\otimes I_n
\end{equation*}
and hence $\mathcal D_A \hat \oplus \mathcal D_B \subseteq \mathcal D_{(A \otimes I_{\mathcal H_2}, I_{\mathcal H_1} \otimes B)}$.
\end{proof}

Let us now justify why we have named the object in Definition \ref{defi:mc_direct_sum} a direct sum. Before we start, we need a lemma.
\begin{lem} \label{lem:polar-level-one}
Let $\mathcal H$ be a Hilbert space and let $A \in (\mathcal B(\mathcal H)^{sa})^g$, $g \in \mathbb N$. Then, 
\begin{equation*}
(\mathcal D_A(1))^\circ =\mathcal W_1(A).
\end{equation*}
\end{lem} 
\begin{proof}
By definition, $x \in \mathcal D_A(1)$ if and only if
\begin{equation*}
\sum_{i = 1}^g x_i \Psi(A_i) \leq 1 \qquad \forall \Psi \in \mathrm{UCP}(\mathcal B(\mathcal H), \mathbb C),
\end{equation*}
since in particular the maps $\Psi:\mathcal B(\mathcal H) \to \mathbb C$, $\Psi: X \mapsto \langle \psi, X \psi \rangle$ with $\psi \in \mathcal H$, $\norm{\psi} = 1$, are UCP.
Hence, $\mathcal D_A(1)=\mathcal W_1(A)^\circ$. The assertion follows from the Bipolar Theorem \cite[Theorem IV.1.2]{Barvinok2002} since $\mathcal W_1(A)$ is a closed convex set which contains $0$ in its interior if and only if $\mathcal D_A(1)$ is bounded \cite[Lemma 3.4]{davidson2016dilations}.
\end{proof}
\begin{remark}
In finite dimensions, $\mathcal W_1(A)$ is just the convex hull $\mathcal C$ of $\{(\langle \psi, A_1 \psi \rangle, \ldots, \langle \psi, A_g \psi \rangle): \psi \in \mathcal H, \norm{\psi} = 1\}$. In infinite dimensions, $\mathcal C$ might not be closed and we have to consider $\mathcal W_1(A)$ instead.
\end{remark}
\begin{prop}
Let $\mathcal H_1$, $\mathcal H_2$ be two Hilbert spaces and let $A \in (\mathcal B(\mathcal H_1)^{sa})^{k_1}$, $B \in (\mathcal B(\mathcal H_2)^{sa})^{k_2}$, where $k_1$, $k_2 \in \mathbb N$. Furthermore, let $\mathcal D_A(1)$ and $\mathcal D_B(1)$ be polytopes. Then, 
\begin{equation*}
(\mathcal D_A \hat \oplus \mathcal D_B)(1) = \mathcal D_A(1) \oplus \mathcal D_B(1).
\end{equation*}
\end{prop}
\begin{proof}
It is easy to see from Proposition \ref{prop:direct-sum-FS} that $(x,0)$ and $(0,y)$ are in $(\mathcal D_A \hat \oplus \mathcal D_B)(1)$ for all $x \in \mathcal D_A(1)$ and all $y \in \mathcal D_B(1)$. Thus, $\mathcal D_A(1) \oplus \mathcal D_B(1) \subseteq \mathcal (D_A \hat \oplus \mathcal D_B)(1)$. For the converse, consider $(x,y) \in (\mathcal D_A \hat \oplus \mathcal D_B)(1)$. Then, in particular
\begin{equation*}
\sum_{i = 1}^{k_1} x_i \Psi_1(A_i) + \sum_{j = 1}^{k_2} y_j \Psi_2(B_j) \leq 1 \qquad \forall \Psi_i \in \mathrm{UCP}(\mathcal B(\mathcal H_i), \mathbb C), i \in [2].
\end{equation*}
This can be seen from an application of $\Psi_1 \otimes \Psi_2 \otimes \mathrm{Id}$ to Equation \eqref{eq:direct-sum-FS}. Thus, $(\mathcal D_A \hat \oplus \mathcal D_B)(1)^{\circ} \supseteq \mathcal W_1(A) \times \mathcal W_1(B)$. Taking the polar dual and applying the Bipolar Theorem \cite[Theorem IV.1.2]{Barvinok2002}, we obtain 
\begin{equation*}
(\mathcal D_A \hat \oplus \mathcal D_B)(1) \subseteq (\mathcal W_1(A) \times \mathcal W_1(B))^\circ = \mathcal D_A(1) \oplus \mathcal D_B(1).
\end{equation*}
The equality on the right hand side follows from Lemmas \ref{lem:bremner} and \ref{lem:polar-level-one}.
\end{proof}
\begin{cor}
Let $\mathcal F$ and $\mathcal G$ be closed matrix convex sets  with $0$ in their interior and such that $\mathcal F_1$ and $\mathcal G_1$ are polytopes. Then,
\begin{equation*}
(\mathcal F \hat \oplus \mathcal G)_1 = \mathcal F_1 \oplus \mathcal G_1.
\end{equation*}
\end{cor}
\begin{proof}
Proposition 3.5 of \cite{davidson2016dilations} shows that there are $g_1$, $g_2 \in \mathbb N$, Hilbert spaces $\mathcal H_1$, $\mathcal H_2$ and $A \in (\mathcal B(\mathcal H_1)^{sa})^{g_1}$, $B \in (\mathcal B(\mathcal H_2)^{sa})^{g_2}$ such that
\begin{equation*}
\mathcal F = \mathcal D_A \qquad \mathrm{and} \qquad \mathcal G = \mathcal D_B.
\end{equation*}
The assertion follows from Proposition \ref{prop:direct-sum-FS}.
\end{proof}

The direct sum we have defined behaves nicely with respect to the maximal spectrahedra for polytopes, as the next lemma shows.

\begin{lem}\label{lem:wmaxpolytopes}
	Let $\mathcal P_1$, $\mathcal P_2$ be two polytopes such that $0 \in \mathrm{int}(\mathcal P_i)$, $i \in [2]$. Then $\mathcal W_{max}(\mathcal P_1 \oplus P_2) = \mathcal W_{max}(\mathcal P_1) \hat \oplus \mathcal W_{max}(\mathcal P_2)$.
\end{lem}
\begin{proof}
	By a refined version of the Weyl-Minkowski theorem, \cite[Lemma VI.1.5]{Barvinok2002}, there exist $c_s^{(i)} \in \mathbb R^{k_i}$, $\alpha^{(i)}_s \in \mathbb R$ such that 
	\begin{equation*}
	\mathcal P_i = \Set{x \in \mathbb R^{k_i}: \langle c_s^{(i)}, x\rangle \leq \alpha_s^{(i)} \quad \forall s \in [m_i]},
	\end{equation*}
	where $m_i \in \mathbb N$. Furthermore, $\mathcal F_s^{(i)} = \Set{p_i \in \mathcal P_i : \langle c_s^{(i)}, p_i\rangle = \alpha_s^{(i)}}$ are the facets of $\mathcal P_i$. By assumption, $0 \in \mathrm{int}(\mathcal P_i)$, and thus $\alpha_s^{(i)} > 0$. Therefore, we can write
	\begin{equation*}
	\mathcal P_i = \Set{x \in \mathbb R^{k_i}: \sum_{j = 1}^{k_i} x_j P_j^{(i)} \leq I_{m_i}} = \Set{x \in \mathbb R^{k_i}: \langle h_s^{(i)}, x\rangle \leq 1 \quad \forall s \in [m_i]},
	\end{equation*}
	where $h_s^{(i)} = c_s^{(i)} / \alpha_s^{(i)}$ and $P_j^{(i)} \in \mathbb R^{m_i}$ such that $P_j^{(i)}(s) = h_s^{(i)}(j)$. As indicated in the beginning of Section \ref{sec:direct-sum}, we identify here vectors in $\mathbb R^{m_i}$ with diagonal $m_i \times m_i$-matrices. Combining Lemma \ref{lem:cartesianfaces} and the fact that facets of a polytope correspond to extreme points of its polar \cite[Theorem VI.1.3]{Barvinok2002}, we find that the extreme points of $\mathcal P_1^\circ \times \mathcal P_2^\circ$ are $(h_{s_1}^{(1)},h_{s_2}^{(2)})$, $s_i \in [m_i]$, $i \in [2]$. Using Lemma \ref{lem:bremner} and Lemma \ref{lem:polarextremepoints}, we obtain
	\begin{equation*}
	\mathcal P_1 \oplus \mathcal P_2 = \Set{(x_1,x_2) \in \mathbb R^{k_1} \times \mathbb R^{k_2}: \langle (h_{s_1}^{(1)},h_{s_2}^{(2)}), (x_1,x_2) \rangle \leq 1 \quad \forall s_i \in [m_i],\, i=1,2}.
	\end{equation*}
	Thus, we find that the $(h_{s_1}^{(1)},h_{s_2}^{(2)})$ are the hyperplanes defining $\mathcal P_1 \oplus \mathcal P_2$. Moreover, we can again write this in spectrahedral form, 
	\begin{equation*}
	\mathcal P_1 \oplus \mathcal P_2 = \Set{x \in \mathbb R^{k_1 + k_2}: \sum_{j = 1}^{k_1 + k_2} x_j Q_j \leq I_{m_1 m_2}}.
	\end{equation*}
	Here, $Q_j \in \mathbb R^{m_1 m_2}$, where $Q_j(s_1,s_2) := (h_{s_1}^{(1)},h_{s_2}^{(2)})_j$. Hence, by the definition of the maximal spectrahedron,
	\begin{equation*}
	\mathcal W_{max}(\mathcal P_1 \oplus \mathcal P_2)(n) = \Set{X \in (\mathcal M_n^{sa})^{k_1 + k_2}: \sum_{j = 1}^{k_1 + k_2} Q_j \otimes X_j \leq I_{n m_1 m_2}}.
	\end{equation*}
	
	Evaluating the expression for the $Q_j$ further, we infer
	\begin{align*}
	Q_j(s_1,s_2) &= \begin{cases} h_{s_1}^{(1)}(j)= P_j^{(1)}(s_1) & 1 \leq j \leq k_1 \\h_{s_2}^{(2)}(j-k_1)= P_{j-k_1}^{(2)}(s_2) & k_1 + 1 \leq j \leq k_1+k_2\end{cases}\\ & = \begin{cases} (P_j^{(1)} \otimes I_{k_2})(s_1,s_2) & 1 \leq j \leq k_1 \\ (I_{k_1} \otimes P^{(2)}_{j-k_1})(s_1,s_2) & k_1 + 1 \leq j \leq k_1+k_2 \end{cases}.
	\end{align*}
	This proves the assertion.
\end{proof}
\begin{remark}
	The assumption $0 \in \mathrm{int}(\mathcal{P})$ is needed to ensure that the polytope $\mathcal P$ can be written as a linear matrix inequality with the identity matrix on the right hand side as in Equation \eqref{eq:lmi}. 
\end{remark}

The next result shows that level-$1$ inclusion of the direct sum of two polytopes \emph{into} a spectrahedron amounts to individual inclusion of each polytope into the corresponding part of the spectrahedron. 

\begin{lem} \label{lem:sumfallsapart}
	Let $A^{(i)} \in (\mathcal M_d^{sa})^{k_i}$, $k_i \in \mathbb N$, $i=1,2$ be two tuples of matrices and $\mathcal P_j \subset \mathbb R^{k_j}$, $j=1,2$ two polytopes. Then,
	\begin{equation*}
	\mathcal P_1 \oplus \mathcal P_2 \subseteq \mathcal D_{(A^{(1)},A^{(2)})}(1) \iff \mathcal P_i \subseteq \mathcal D_{A^{(i)}}(1) \quad i=1,2. 
	\end{equation*}
\end{lem}
\begin{proof}
	Let $\Set{w_j^{(i)}}_{j = 1}^{m_i} \subset \mathbb R^{k_i}$ be the set of extreme points of $\mathcal P_i$ with $m_i \in \mathbb N$. Then, the set of extreme points of $\mathcal P_1 \oplus \mathcal P_2$ is $\Set{(w_{j_1}^{(1)},0), (0,w^{(2)}_{j_2}): j_i \in [m_i],\, i=1,2}$. This can easily be seen from the definition. Since the inclusion of polytopes can be checked at the extreme points, the assertion follows. 
\end{proof}

To finish this Section, let us compare the Cartesian product and direct sum of matrix convex sets we have defined to each other and other constructions in the literature.
\begin{remark}\label{rem:different-direct-sums}
As pointed out earlier, our definition of Cartesian product coincides with Definition 4.1 of \cite{Passer2019}. Proposition 4.5 of \cite{Passer2019} shows that for closed and bounded convex sets $\mathcal C$, $\mathcal D$ in $\mathbb R^{g_1}$ and $\mathbb R^{g_2}$ respectively, 
\begin{equation*}
\mathcal W_{max}(\mathcal C \times \mathcal D) = \mathcal W_{max}(\mathcal C)\hat \times \mathcal W_{max}(\mathcal D). 
\end{equation*}
Thus,
\begin{equation*}
\mathcal W_{max}([-1,1])\hat \times \mathcal W_{max}([-1,1]) = \mathcal W_{max}([-1,1]^2) = \mathcal D_{\square, 2},
\end{equation*}
where $\mathcal D_{\square, 2}$ is the complex matrix cube \cite[Example 2.3]{davidson2016dilations} (see also Equation \eqref{eq:matrix-cube}). For the direct sum, Lemma \ref{lem:wmaxpolytopes} implies
\begin{equation*}
\mathcal W_{max}([-1,1])\hat \oplus \mathcal W_{max}([-1,1]) = \mathcal W_{max}(\mathcal B_1(\mathbb R^2))=\mathcal D_{\diamond, 2}
\end{equation*}
Here, $\mathcal B_1(\mathbb R^2)$ is the $\ell_1$-ball in $\mathbb R^2$ and $\mathcal D_{\diamond, 2}$ the matrix diamond \cite[Section 10.3]{davidson2016dilations}. We see that the matrix diamond and the matrix cube differ only with respect to the operation used to construct a new matrix convex set from copies of $\mathcal W_{max}([-1,1])$. The paper \cite{Passer2019} considers yet another operation on matrix convex sets in Definition 4.1 which the authors call $\times_1$. For matrix convex sets $\mathcal F$ and $\mathcal G$, also $(\mathcal F \times_1 \mathcal G)_1 = \mathcal F_1 \oplus \mathcal G_1$ holds, such that it generalizes the direct sum of convex sets. However, \cite[Proposition 4.5]{Passer2019} shows that for closed and bounded matrix convex sets,
\begin{equation*}
\mathcal W_{min}(\mathcal F_1 \oplus \mathcal G_1) = \mathcal W_{min}(\mathcal F_1) \times_1 \mathcal W_{min}(\mathcal G_1),
\end{equation*}
where $\mathcal W_{min}(\mathcal F_1)$ is the minimal matrix convex set with $\mathcal F_1$ at the first level (see \cite[Section 4]{davidson2016dilations} for details). By \cite[Section 4]{passer2018minimal}, $\mathcal W_{min}(\mathcal C) = \mathcal W_{max}(\mathcal C)$ for a compact convex set $\mathcal C \subset \mathbb R^g$ if and only if $\mathcal C$ is a simplex. Thus,
\begin{equation*}
\mathcal W_{max}([-1,1])\hat \times_1 \mathcal W_{max}([-1,1]) = \mathcal W_{min}(\mathcal B_1(\mathbb R^2)),
\end{equation*}
because the matrix convex set with $[-1,1]$ at the first level is unique. Since $\mathcal B_1(\mathbb R^2)$ is not a simplex, $\mathcal W_{min}(\mathcal B_1(\mathbb R^2)) \neq \mathcal D_{\diamond, 2}$ and we find that $\hat \oplus$ and $\times_1$ are different operations in general. We remark that $\times$ and $\oplus$ are dual for usual convex sets are dual operations, but $\hat \times$ and $\hat \oplus$ do not give rise to dual matrix convex sets. 
\end{remark}

\subsection{Quantum information theory}

We will conclude this section with a short review of some concepts from quantum information theory which we will use. For an introduction to the mathematics of quantum mechanics, see e.g.\ \cite{Heinosaari2011} or \cite{watrous2018theory}.
A quantum mechanical system is given as a \emph{state} $\rho \in \mathcal S(\mathcal H)$. Here, $\mathcal H$ is the Hilbert space of the system and
\begin{equation*}
\mathcal S(\mathcal H):= \Set{\rho \in \mathcal B(\mathcal H): \rho \geq 0, \tr[\rho] =1}.
\end{equation*}
In the present work, we will only deal with finite-dimensional Hilbert spaces. A state is \emph{pure} if it has rank one. Valid transformations between quantum systems are given in terms of completely positive maps. Let $\mathcal H$, $\mathcal K$ be two Hilbert spaces and $\mathcal T: \mathcal B(\mathcal H) \to \mathcal B(\mathcal K)$ be a linear map. This map is $k$-positive if the map $\mathcal T \otimes \mathrm{Id}_k: \mathcal B(\mathcal H) \otimes \mathcal M_k \to \mathcal B(\mathcal K) \otimes \mathcal M_k$ is positive for $k \in \mathbb N$. It is completely positive if $\mathcal T$ is $k$-positive for all $k \in \mathbb N$. For $\mathcal T$ to be a \emph{quantum channel}, we require additionally that the map is trace preserving. In finite dimensions where $d:= \dim(\mathcal K) < \infty$, $d$-positivity of $\mathcal T$ is equivalent to complete positivity \cite[Theorem 6.1]{Paulsen2002}.

Quantum mechanical measurements are described using effect operators, i.e.\
\begin{equation*}
\mathrm{Eff}_d := \Set{E \in \mathcal M_d^{sa}: 0 \leq E \leq I}.
\end{equation*} 
A measurement then corresponds to a \emph{positive operator valued measure} (POVM). Let $\Sigma$ be the set of measurement outcomes, which we assume to be finite for simplicity. The corresponding POVM is then a set of effects $\Set{E_j}_{j \in \Sigma}$, $E_j \in \mathrm{Eff}_d$ for all $j \in \Sigma$, such that
\begin{equation*}
\sum_{j \in \Sigma} E_j = I_d.
\end{equation*}
Since the actual measurement outcomes are not important for us, we will write $\Sigma = [k]$ for some $k \in \mathbb N$.

The main concept for the rest of this work is the notion of joint measurability. A collection of POVMs is jointly measurable if they arise as marginals from a joint POVM (see \cite{Heinosaari2016} for an introduction). 
\begin{defi}[Jointly measurable POVMs] \label{def:jointPOVM}
Let $\Set{E_j^{(i)}}_{j \in [k_i]}$ be a collection of $d$-dimensional POVMs, where $k_i \in \mathbb N$ for all $i \in [g]$, $g \in \mathbb N$. The POVMs are \emph{jointly measurable} (often also called \emph{compatible}) if there is a $d$-dimensional joint POVM $\Set{R_{j_1, \ldots, j_g}}$ with $j_i \in [k_i]$ such that for all $u \in [g]$ and $v \in [k_u]$,
\begin{equation*}
E_v^{(u)} = \sum_{\substack{j_i \in [k_i] \\ i \in [g] \setminus \Set{u}}} R_{j_1, \ldots, j_{u-1},v,j_{u+1}, \ldots j_g}.
\end{equation*}
\end{defi}

There is an equivalent definition of joint measurability \cite[Equation 16]{Heinosaari2016}, formulated in terms of post-processing, which will sometimes be useful. Measurements are compatible if and only if they arise through post-processing from a common measurement. 
\begin{lem} \label{lem:post-processing}
Let $E^{(i)} \in (\mathcal M_d^{sa})^{k_i}$, $i \in [g]$, be a collection of POVMs. These POVMs are jointly measurable if and only if there is some $m \in \mathbb N$ and a POVM $M \in (\mathcal M_d^{sa})^{m}$ such that
\begin{equation*}
E^{(i)}_{j} = \sum_{x=1}^m p_i(j|x)M_x
\end{equation*}
for all $j \in [k_i]$, $i \in [g]$ and some conditional probabilities $p_i(j|x)$.
\end{lem}

Not all measurements in quantum mechanics are compatible, but they can be made compatible if we add enough noise. By adding noise we mean taking the convex combination of a POVM and a trivial measurement, i.e\ a POVM in which all effects are proportional to the identity. These are called trivial, because they do not depend on the state of the system. With this idea, we can define several compatibility regions, i.e.\ sets of noise parameters for which any collection of a fixed number of measurements in fixed dimension and with a fixed number outcomes is compatible. For the first such set, we restrict to balanced noise.

\begin{defi}
Let $\mathbf k \in \mathbb N^g$, $d$, $g \in \mathbb N$. Then, we call
\begin{equation*}
\Gamma(g,d, \mathbf k) := \Set{s \in [0,1]^g: s_{i} E^{({i})} + (1-s_{i})I/k_{i} \mathrm{~compatible~}\forall \mathrm{~POVMs~}E^{({i})} \in (\mathcal M_d^{sa})^{k_{i}}}
\end{equation*}
the \emph{balanced compatibility region} for $g$ POVMs in $d$ dimensions with $k_i$ outcomes, $i \in [g]$.
\end{defi}

Sometimes it is desirable that the noise is linear in the effect operators. Such noise arises in the framework of quantum steering \cite{Uola2014, Heinosaari2015}.
\begin{defi}
Let $\mathbf k \in \mathbb N^g$, $d$, $g \in \mathbb N$. Then, we call
\begin{equation*}
\Gamma^{lin}(g,d, \mathbf k) := \Set{s \in [0,1]^g: \left[ s_i E^{(i)}_j + (1-s_i)\frac{\tr[E^{(i)}_j]}{d}I \right]_{j \in [k_i]} \!\!\!\!\!\!\!\!\!\!\!\!\!\!\!\mathrm{~compatible~}\forall \mathrm{~POVMs~}E^{(i)} \in (\mathcal M_d^{sa})^{k_i}}
\end{equation*}
the \emph{linear compatibility region} for $g$ POVMs in $d$ dimensions with $k_i$ outcomes, $i \in [g]$.
\end{defi}

Let us prove a lemma which shows that coarse graining, i.e.\ grouping several outcomes together, does not destroy joint measurability. 
\begin{lem}\label{lem:bunch-together}
Let $E^{(i)} \in (\mathcal M_d^{sa})^{k_i^\prime}$, $k_i^\prime \in \mathbb N$, $i \in [g]$, be a collection of jointly measurable POVMs. Then, also $E^{(i)}$, $i \in [g] \setminus \Set{l}$ and $\tilde E^{(l)}$ are jointly measurable, where
\begin{equation*}
\tilde E^{(l)} = (E^{(l)}_1, \ldots, E^{(l)}_{k_l}, E^{(l)}_{k_l +1} + \ldots + E^{(l)}_{k_l^\prime})
\end{equation*}
and $l \in [g]$, $k_l \in \mathbb N$, $k_l \leq k_l^\prime$.
\end{lem}
\begin{proof}
Let $G_{j_1, \ldots, j_g}$, $j_i \in k_i^\prime$, $i \in [g]$ be a joint POVM for the $E^{(i)}$. Then, we can define a new POVM as
\begin{equation*}
\tilde G_{j_1, \ldots, j_g} = \begin{cases} G_{j_1, \ldots, j_g} &j_l \leq k_l \\
\sum_{r = k_l+1}^{k_l^\prime} G_{j_1, \ldots, j_{l-1}, r , j_{l+1},\ldots, j_g} & j_l = k_l+1 \end{cases}.
\end{equation*}
Note that on the left hand side, $j_i \in k^\prime_i$ for $i \in [g]\setminus \Set{l}$ and $j_l \in [k_l + 1]$. It can easily be verified that this POVM is a joint POVM for the $E^{(i)}$  (with $i \neq l$) and $\tilde E^{(l)}$.
\end{proof}

\begin{prop} \label{prop:different-k-gamma-inclusion}
Consider two $g$-tuples of positive integers $\mathbf k, \mathbf{k'}$ such that $\mathbf k^\prime \geq \mathbf k$ (coordinate-wise, i.e.~$k'_i \geq k_i, \, \forall i \in [g]$). Let $\# \in \Set{\emptyset, lin}$. Then,
\begin{equation*}
\Gamma^{\#}(g,d,\mathbf k^\prime) \subseteq \Gamma^{\#}(g,d,\mathbf k).
\end{equation*}
\end{prop}
\begin{proof}
Fix $s \in \Gamma(g,d,\mathbf k^\prime)$. Let furthermore $E^{(i)} \in (\mathcal M_d^{sa})^{k_i}$, $i \in [g]$ be a collection of POVMs. Let $\tilde E^{(i)} \in (\mathcal M_d^{sa})^{k^\prime_i}$ be the POVM which is equal to $E^{(i)}$ in the first $k_i$ entries and $0$ for the rest. Then, since $s \in \Gamma(g,d,\mathbf k^\prime)$, the POVMs $s_i \tilde E^{(i)} + (1-s_i) I/k_i^\prime$ are jointly measurable. Let
\begin{align*}
F^{(i)} &= \left(s_i \tilde E_{1}^{(i)} + (1-s_i) I/k_{i}^\prime, \ldots, s_i \tilde E_{k_{i}}^{(i)} + (1-s_i) I/k_{i}^\prime, (1-s_i)\frac{k^\prime_{i} -k_{i}}{{k_i}^\prime}I\right)\\
&=\left(s_i  E_{1}^{(i)} + (1-s_i) I/k_{i}^\prime, \ldots, s_i  E_{k_{i}}^{(i)} + (1-s_i) I/k_{i}^\prime, (1-s_i)\frac{k^\prime_{i} -k_{i}}{{k_i}^\prime}I\right).
\end{align*}
An iterative application of Lemma \ref{lem:bunch-together} shows that also the $F^{(i)}$ are jointly measurable with joint POVM $G$. Let $\mathbf{j} = (j_1, \ldots, j_g) \in [k_1+1] \times \ldots \times [k_g+1]$. Define, for $i \in [g]$ and $l \in [k_i]$,
\begin{equation*}
p_i(l|\mathbf{j}) = \begin{cases} 
1 & j_i = l \\
\frac{1}{k_i} & j_i=k_i+1\\ 
0 & \mathrm{else} 
\end{cases}.
\end{equation*}
These are conditional probabilities and it holds that
\begin{align*}
\sum_{\mathbf{j} \in \times_{i=1}^g[k_i+1]} p_i(l|\mathbf{j}) G_{\mathbf{j}} &= \sum_{\substack{\mathbf{j} \in \times_{i=1}^g[k_i+1]\\{j_i} = l}} G_{\mathbf{i}} + \frac{1}{k_i} \sum_{\substack{\mathbf{j} \in \times_{i=1}^g[k_i+1]\\{j_i} = k_i+1}} G_{\mathbf{j}} \\
&= F^{(i)}_l + \frac{1}{k_i} F^{(i)}_{k_i+1}\\
&= s_i E^{(i)}_{l} + (1-s_i) \frac{I}{k_i^\prime} + \frac{1}{k_i} \frac{k_i^\prime - k_i}{k_i^\prime} (1-s_i)I \\
& = s_i E^{(i)}_l + (1-s_i) \frac{I}{k_i}.
\end{align*}
From Lemma \ref{lem:post-processing}, it follows that $s \in \Gamma(g,d,\mathbf k)$. The assertion for $\Gamma^{lin}$ follows directly from extending the POVMs by zeroes.

For the second assertion, choose $s \in \Gamma^{lin}(g,d,\mathbf{k}^\prime)$ and a collection of POVMs $E^{(i)} \in (\mathcal M_d^{sa})^{k_i}$, $i \in [g]$. Let again $\tilde E^{(i)} \in (\mathcal M_d^{sa})^{k^\prime_i}$ be the POVM which is equal to $E^{(i)}$ in the first $k_i$ entries and $0$ for the rest. From the choice of $s$, it follows that the POVMs 
\begin{align*}
&\left[s_i \tilde E^{(i)}_j + (1-s_i) \frac{\tr[\tilde E^{(i)}_j]}{d} I\right ]_{j \in [k_i^\prime]} \\
=& \left(s_i  E^{(i)}_1 + (1-s_i) \frac{\tr[ E^{(i)}_1]}{d} I, \ldots, s_i  E^{(i)}_{k_i} + (1-s_i) \frac{\tr[ E^{(i)}_{k_i}]}{d} I, 0, \ldots, 0 \right)
\end{align*}
are compatible with joint POVM $G_{\mathbf j}$, where $\mathbf j \in [k^\prime_1] \times \ldots \times [k^\prime_g]$. As $G_{\mathbf j} \geq 0$ for all $\mathbf j$, it follows that $G_{\mathbf j} = 0$ if $j_i \in [k_i^\prime]\setminus [k_i]$ for some $i \in [g]$, since these elements have to sum up to $0$ by Definition \ref{def:jointPOVM}. Therefore, $\left[ G_{\mathbf{j}} \right]_{\mathbf j \in [k_1] \times \ldots \times [k_g]}$ is still a POVM and moreover a joint POVM for the $E^{(i)}$.  This shows that $s \in \Gamma^{lin}(g,d,\mathbf{k})$.
\end{proof}

The following proposition generalizes \cite[Proposition III.4(6)]{bluhm2018joint}.

\begin{prop} \label{prop:lin-gamma-inclusion}
Let $\mathbf k \in \mathbb N^g$. Furthermore, let $k_{max}=\max_{i \in [g]} k_i$. Then,
\begin{equation*}
\Gamma^{lin}(g,k_{max}d, k_{max}^{\times g}) \subseteq \Gamma(g,d,\mathbf k).
\end{equation*}
\end{prop}
\begin{proof}
From Proposition \ref{prop:different-k-gamma-inclusion}, it follows that
\begin{equation*}
\Gamma(g,d,k_{max}^{ \times g}) \subseteq \Gamma(g,d,\mathbf k),
\end{equation*}
so it is enough to prove 
\begin{equation*}
\Gamma^{lin}(g,k_{max}d, k_{max}^{\times g}) \subseteq \Gamma(g,d,k_{max}^{\times g}).
\end{equation*}
Pick $g$ POVMs $E^{(i)}$ of dimension $d$ and with $k_{max}$ outcomes each, $i \in [g]$. Let
\begin{equation*}
F_j^{(i)} = E_j^{(i)} \oplus E_{j+1}^{(i)} \oplus \ldots \oplus E_{j+(k_{max}-1)}^{(i)} \qquad \forall j \in [k_{max}], \forall i \in [g].
\end{equation*}
Above, we are considering the addition operation modulo $k_{max}$, i.e.~ we identify $i + k_{max}$ with $i$ for $i \in k_{max}$. Thus, $F_j^{(i)} \in \mathcal M_{d k_{max}}$ for all $j \in [k_{max}]$, $\forall i \in [g]$. Clearly, $F_j^{(i)} \geq 0$ and $\sum_{j = 1}^{k_{max}} F_j^{(i)} = I_{d k_{max}}$ for any $i \in [g]$, so the $F^{(i)}$ again are POVMs. Let $s \in \Gamma^{lin}(g,k_{max}d, k_{max}^{\times g})$. Then, the $s_i F^{(i)} + (1-s_i) I_{ k_{max}d}/k_{max}$ are jointly measurable POVMs, because $\tr[F_j^{(i)}]/(k_{max}d) = 1/k_{max}$.
Applying an isometry onto the first block of the direct sum ascertains that the $s_i E^{(i)} + (1-s_i)I_d/k_{max}$ are jointly measurable as well. Since the POVMs we picked were arbitrary, the assertion follows.  
\end{proof}

\section{The matrix jewel}\label{sec:jewel}

In the following, we identify the subalgebra of $d \times d$ diagonal matrices with $\mathbb C^d$.
\begin{defi}[Matrix jewel] \label{def:jewel} 
Consider the vectors $v^{(k)}_1, \ldots, v^{(k)}_{k-1} \in \mathbb C^{k}$ defined as
\begin{equation*}
v^{(k)}_j(\epsilon) := -\frac{2}{k} + 2\delta_{\epsilon, j}, \qquad \forall j \in [k-1],\, \forall \epsilon \in [k].
\end{equation*}
The free spectrahedron $\mathcal D_{\jewel,k}$ defined by
\begin{equation*}
\mathcal D_{\jewel,k}(n):= \Set{X \in (\mathcal M_n^{sa})^{k-1}: \sum_{j = 1}^{k-1} v^{(k)}_j \otimes X_j \leq I_{kn}}. 
\end{equation*}
is called the \emph{matrix jewel base}. For a $g$-tuple of positive integers $\mathbf k = (k_1, \ldots, k_g)$, we define the \emph{matrix jewel} $\mathcal D_{\jewel, \mathbf k}$ to be the free spectrahedron
	$$\mathcal D_{\jewel, \mathbf k} := \mathcal D_{\jewel, k_1} \hat \oplus \mathcal D_{\jewel, k_2} \hat \oplus \cdots \hat \oplus \mathcal D_{\jewel, k_g},$$
where the direct sum operation $\hat \oplus$ for free spectrahedra was introduced in Section \ref{sec:direct-sum}. In other words, we have
\begin{equation}\label{eq:jewel-explicit}
\mathcal D_{\jewel, \mathbf k}(n) = \left\{X \in (\mathcal M_n^{sa})^{\sum_{i=1}^g (k_i-1)}: \sum_{i=1}^g \sum_{j = 1}^{k_i-1} \left[I^{\otimes (i-1)} \otimes v^{(k_i)}_j \otimes I^{\otimes (g-i)}\right] \otimes  X_{i,j} \leq I_{(\prod_{s=1}^gk_i)n} \right\}.
\end{equation}
\end{defi}

\begin{remark}
It follows immediately from Lemma \ref{lem:wmaxpolytopes} that the matrix jewel is the maximal matrix convex set (in the sense of \cite[Section 4]{davidson2016dilations}, see also Equation \eqref{eq:Wmax}) built on top of the direct sum of simplices
$$\mathcal D_{\jewel, k_1}(1)  \oplus \mathcal D_{\jewel, k_2}(1) \oplus \cdots  \oplus \mathcal D_{\jewel, k_g}(1).$$
\end{remark}

At level one, the matrix jewel base is isomorphic to a simplex, for which we can identify the extremal points. 

\begin{lem}\label{lem:extremal-points-jewel-base}
	The extremal points of the jewel base $\mathcal D_{\jewel, k}(1)\subseteq \mathbb R^{k-1}$ are
	\begin{align*}
	x_i^{(k)} &:= -\frac{k}{2} e_i, \qquad\qquad \text{ for } i \in [k-1]\\
	x_k^{(k)} &:= \frac{k}{2} (\underbrace{1, \ldots, 1}_{k-1 \text{ times}}),
	\end{align*} 
	where $e_i$ are the elements of the standard orthonormal basis in $\mathbb R^{k-1}$.
\end{lem}
\begin{proof}
	Since $\mathcal D_{\jewel, k}(1)$ is a polyhedron and since the hyperplanes $(v_1(\epsilon), \ldots v_{k-1}(\epsilon))_{\epsilon=1}^k$ are such that each $k-1$ of them linearly span $\mathbb R^{k-1}$, \cite[Theorem II.4.2]{Barvinok2002} implies that it is enough to check whether each point as above fulfills $k-1$ of the above constraints with equality (there is no point which fulfills all constraints with equality). We verify for fixed $\epsilon \in [k]$:
	\begin{equation*}
	\sum_{j = 1}^{k-1} v_j(\epsilon)(-\frac{k}{2} e_i)_j = 1 - k \delta_{\epsilon,i}, \qquad i \in [k-1],
	\end{equation*}
	and 
	\begin{equation*}
	\sum_{j = 1}^{k-1} v_j(\epsilon)\frac{k}{2}(1, \ldots, 1)_j = 1-k \delta_{\epsilon,k},
	\end{equation*}
	which proves the claim. 
\end{proof}

At level 1, the matrix jewel base is, for $k=2$, the segment $[-1,1] \subseteq \mathbb R$. We display in Figure \ref{fig:matrix-jewel-base} the sets $\mathcal D_{\jewel,k}(1)$, for $k=3,4$. 
\begin{figure}[htbp]
	\begin{center}
		\includegraphics[scale=.53]{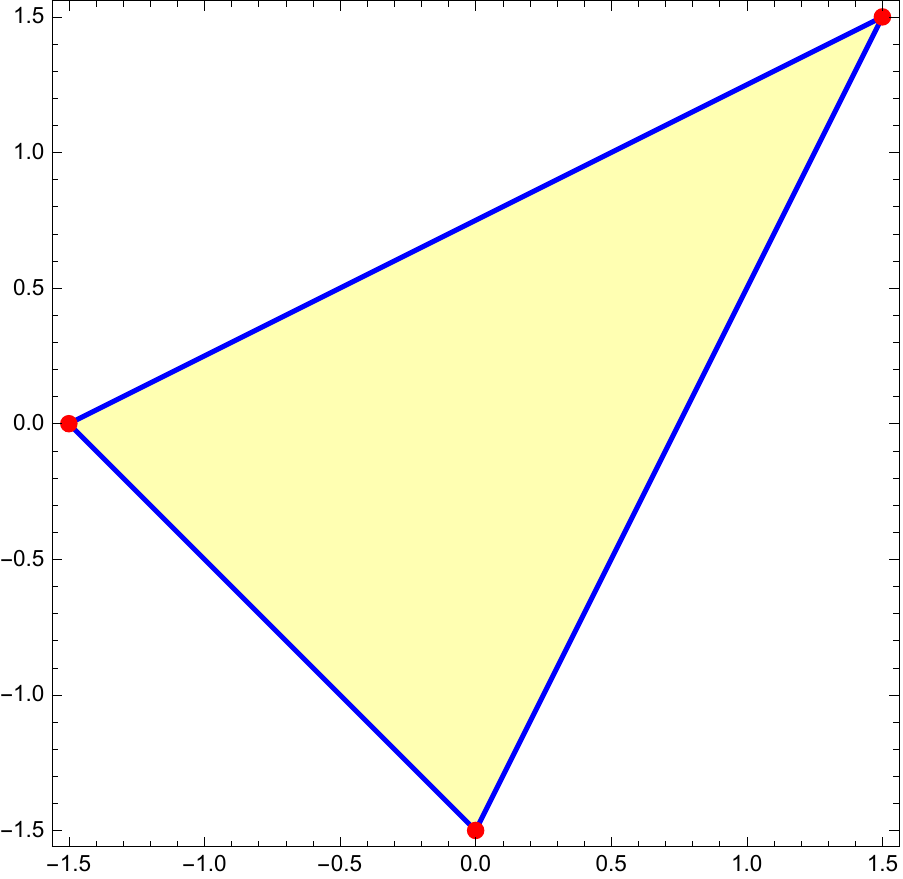}\qquad\qquad\qquad  \includegraphics[scale=.53]{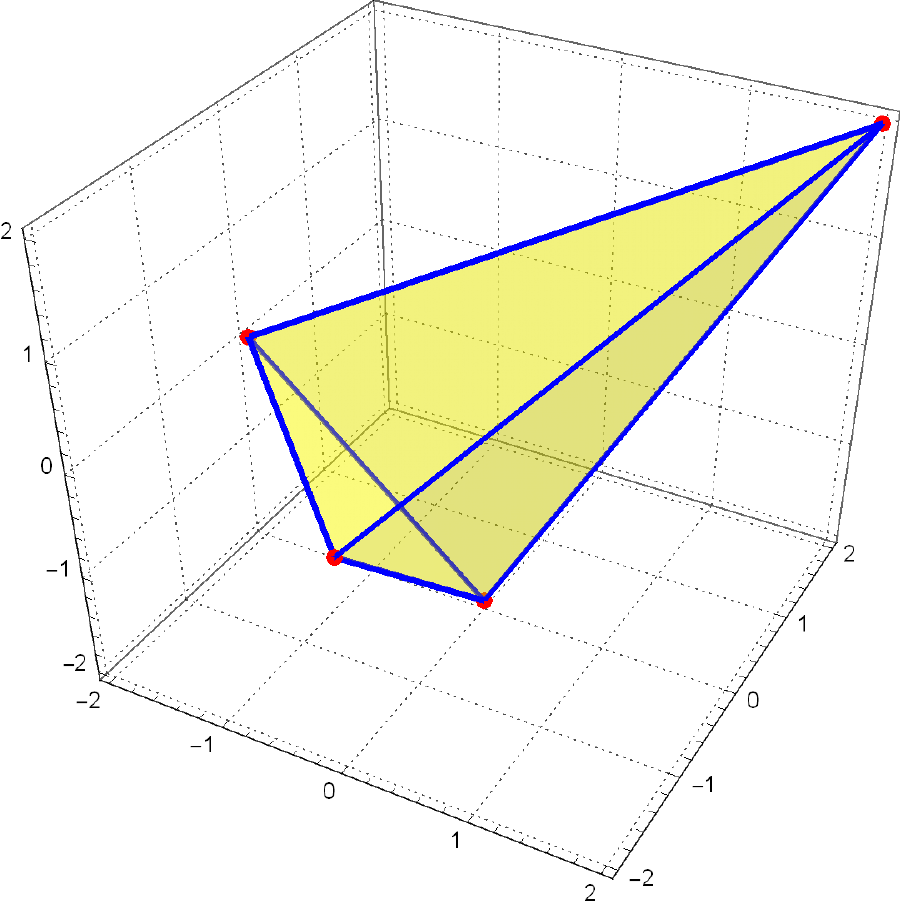}
		\caption{The spectrahedron level of the matrix jewel base $\mathcal D_{\jewel,k}(1)$, for $k=3,4$.}
		\label{fig:matrix-jewel-base}
	\end{center}
\end{figure}
The notion of matrix jewel generalizes the matrix diamond introduced in \cite{davidson2016dilations}; indeed, with the notation of \cite{bluhm2018joint}, the matrix diamond of size $g$ is given by
	$$\mathcal D_{\diamondsuit,g} = \mathcal D_{\jewel, (\underbrace{2,\ldots, 2}_{g \text{ times}})} = \widehat \bigoplus_{i=1}^g \mathcal D_{\jewel,2}.$$
	In Figure \ref{fig:matrix-jewel}, we print the first level of the matrix jewel, for vectors $\mathbf k$ equal to, respectively, $(2,2)$, $(2,2,2)$, and $(2,3)$.

\begin{figure}[htbp]
	\begin{center}
		\includegraphics[scale=.53]{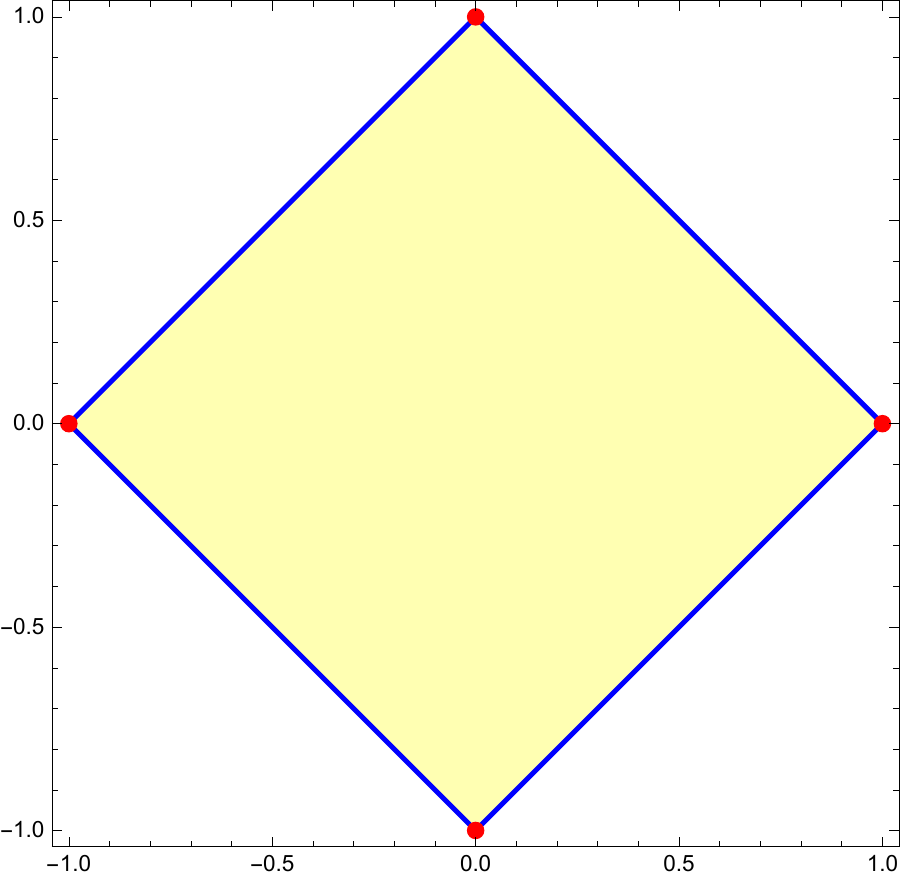}\qquad  \includegraphics[scale=.53]{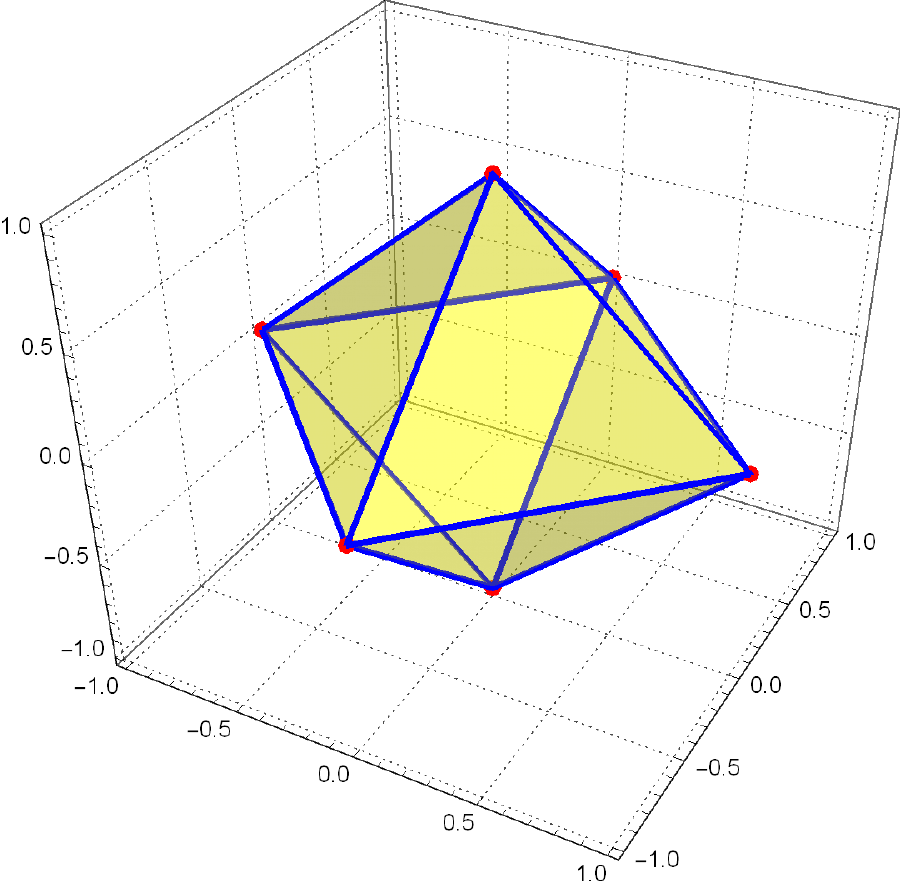}\qquad  \includegraphics[scale=.53]{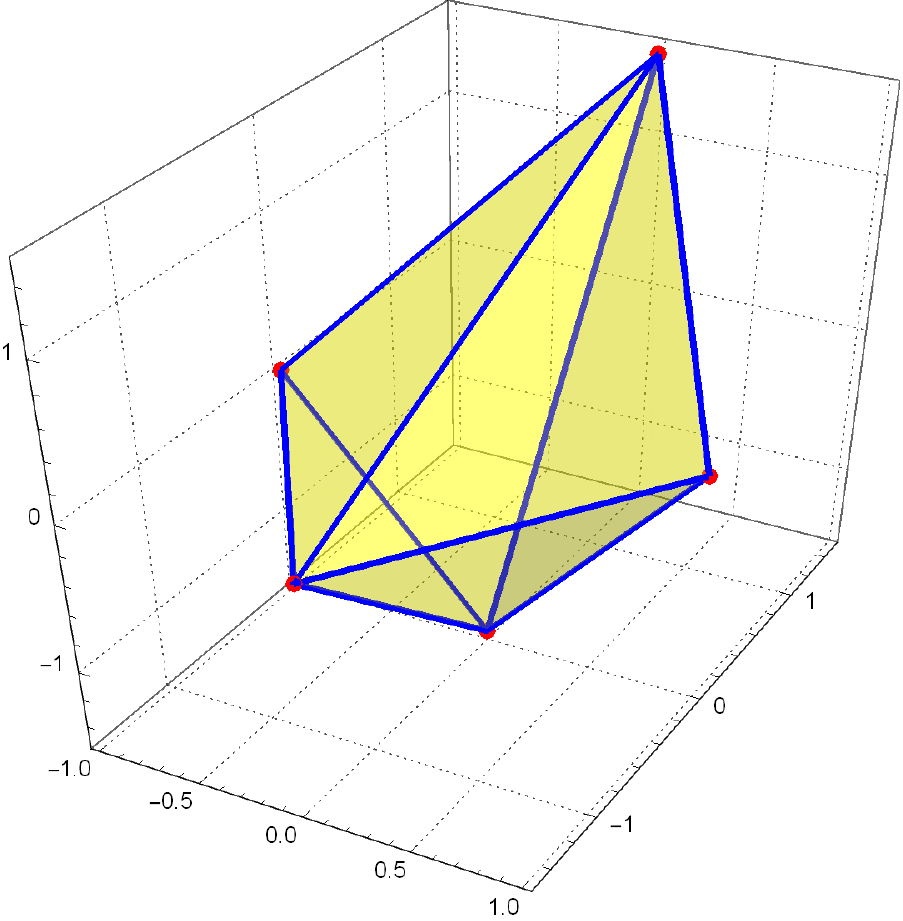}		\caption{The spectrahedron level of the matrix jewels $\mathcal D_{\jewel,(2,2)}(1)$, $\mathcal D_{\jewel,(2,2,2)}(1)$, and $\mathcal D_{\jewel,(2,3)}(1)$. The first two are in fact the matrix diamonds $\mathcal D_{\diamondsuit, 2}(1)$ and $\mathcal D_{\diamondsuit, 3}(1)$ from \cite{bluhm2018joint} (a square and an octahedron), while the last polyhedron is new.}
		\label{fig:matrix-jewel}
	\end{center}
\end{figure}

\section{The matrix jewel and joint measurability of POVMs} \label{sec:jewel-compatibility}

In this section, we establish an equivalence between the inclusion of the matrix jewel in a spectrahedron defined by a tuple of POVMs and the joint measurability of the POVMs. The inclusion at different levels will correspond to different notions of joint measurability. Our first result relates the inclusion of the matrix jewel base, at level 1, to the definition of a POVM.

\begin{prop} \label{prop:spectrahedronPOVM}
Let $E \in (\mathcal M_d^{sa})^{k-1}$. Then, $\Set{E_1, \ldots, E_{k-1}, I - E_1 - \ldots - E_{k-1}}$ is a POVM if and only if 
\begin{equation*}
\mathcal D_{\jewel, k}(1) \subseteq \mathcal D_{2E - \frac{2}{k}I}(1).
\end{equation*}
\end{prop}
\begin{proof}
Since the left hand side is a polytope, we only need to check the assertion on the extremal points $x^{(k)}_j$ from Lemma \ref{lem:extremal-points-jewel-base}. We have
\begin{equation*}
-\frac{k}{2}e_i \in \mathcal D_{2 E - \frac{2}{k}I}(1) \iff -\frac{k}{2}(2 E_i - \frac{2}{k}I) \leq I \iff E_i \geq 0
\end{equation*}
and 
\begin{equation*}
\frac{k}{2}(1, \ldots, 1) \in \mathcal D_{2E - \frac{2}{k}I}(1) \iff \frac{k}{2} \sum_{i = 1}^{k-1} (2E_i - \frac{2}{k}I) \leq I \iff \sum_{i = 1}^{k-1} E_i \leq I.
\end{equation*}
This proves the assertion.
\end{proof}

The following theorem is one of our main results, connecting joint measurability of arbitrary POVMs to the inclusion of the matrix jewel. It is a generalization of \cite[Theorem V.3]{bluhm2018joint} from the case of $g$ binary (i.e.~2-outcome) POVMs to general POVMs (with an arbitrary number of outcomes). 

\begin{thm}\label{thm:jewel-compatible-POVM}
For a fixed matrix dimension $d$, consider $g$ tuples of self-adjoint matrices $E^{(i)} \in (\mathcal M_{d}^{sa})^{k_i - 1}$, $k_i \in \mathbb N$, $i \in [g]$. Define $E^{(i)}_{k_i}:= I_d - E^{(i)}_1 \ldots - E^{(i)}_{k_i -1}$, set $\mathbf k = (k_1, \ldots, k_g)$, and write
\begin{align*}
\mathcal D_{E}&:=\mathcal D_{(2E^{(1)}-\frac{2}{k_1} I,\ldots,  2E^{(g)}-\frac{2}{k_g} I)}\\& = \bigsqcup_{n=1}^\infty \left\{ X \in (\mathcal M_n^{sa})^{\sum_{i=1}^g (k_i-1)} \, : \, \sum_{i=1}^g \sum_{j=1}^{k_i-1} \left(2E^{(i)}_j-\frac{2}{k_i}I\right)\otimes X_{i,j} \leq I_{dn}\right\}.
\end{align*}
Then
\begin{enumerate}
\item $\mathcal D_{\jewel, \mathbf k}(1)  \subseteq \mathcal D_E(1)$ if and only if $\Set{E^{(i)}_1, \ldots, E^{(i)}_{k_i}}$, $i \in [g]$, are POVMs.
\item $\mathcal D_{\jewel, \mathbf k}  \subseteq \mathcal D_E$ if and only if $\Set{E^{(i)}_1, \ldots, E^{(i)}_{k_i}}$, $i \in [g]$, are jointly measurable POVMs.
\item $\mathcal D_{\jewel, \mathbf k}(l)  \subseteq \mathcal D_E(l)$ for $l \in [d]$ if and only if for any isometry $V : \mathbb C^l \hookrightarrow \mathbb C^d$, the tuples $\Set{V^*E^{(i)}_1 V, \ldots, V^* E^{(i)}_{k_i}V}$, $i \in [g]$, are jointly measurable POVMs.
\end{enumerate}
\end{thm}
\begin{proof}
Since $\mathcal D_{\jewel, k_i}(1)$ is a polytope for all $i \in [g]$ and $D_{\jewel, k_i} = \mathcal W_{max}(D_{\jewel, k_i}(1))$, the first assertion follows from Lemmas \ref{lem:wmaxpolytopes} and \ref{lem:sumfallsapart} together with Proposition \ref{prop:spectrahedronPOVM}. 

For the second assertion, let us define, for $i \in [g]$ and $j \in [k_i-1]$,
$$w^{(i)}_j := \underbrace{I \otimes \cdots \otimes I}_{i-1 \text{ times}} \otimes v^{(k_i)}_j \otimes \underbrace{I \otimes \cdots \otimes I}_{g-i \text{ times}}.$$
Here, the $v_j^{(k_i)}$  are (identified with) the diagonal matrices appearing in Definition \ref{def:jewel}, with the appropriate matrix dimension ($k_i$ in the formula above). The free spectrahedral inclusion holds if and only if the unital map $\Phi: \mathcal O \mathcal S_{\Set{w^{(i)}_j}_{i \in [g], j \in [k_i -1]}} \to \mathcal M_d$, defined as 
\begin{equation*}
\Phi: w_j^{(i)}  \mapsto 2 E^{(i)}_j - \frac{2}{k_i}I \qquad \forall i \in [g], \forall j \in [k_i -1],
\end{equation*}
is completely positive, since $\mathcal D_{\jewel, \mathbf k}(1)$ is a polytope and therefore bounded. By Arveson's extension theorem $ \Phi$ is completely positive if and only if there is a completely positive extension $\tilde\Phi: \mathbb C^{k_1 \cdots k_g} \to \mathcal M_d$ of $\Phi$, as $\Phi$ is defined on an operator system \cite[Theorem 6.2]{Paulsen2002}. As $\mathbb C^{k_1 \cdots k_g}$ is a commutative matrix subalgebra, $\tilde \Phi$ is completely positive if and only if it is positive. Let $\tilde \Phi: \mathbb C^{k_1 \cdots k_g} \to \mathcal M_d$ be an extension $\Phi$. We will show now that $\tilde \Phi$ is positive if and only if the $E^{(i)}$ form a set of jointly measurable POVMs.

Let $\epsilon \in [\mathbf k] := \times_{i=1}^g [k_i]$. Then,
\begin{equation} \label{eq:oselements}
w^{(i)}_j(\epsilon) = - \frac{2}{k_i} + 2 \delta_{\epsilon(i),j}.
\end{equation}
Let $g_\eta \in \mathbb C^{k_1 \cdots k_g}$, $\eta \in [\mathbf k]$ such that $g_\eta(\epsilon) = \delta_{\epsilon,\eta}$. These vectors form a basis of $\mathbb C^{k_1 \cdots k_g}$. Hence, we can rewrite Equation \eqref{eq:oselements} as
\begin{equation*}
w^{(i)}_j(\epsilon)= - \frac{2}{k_i} \sum_{\eta \in [\mathbf k]} g_\eta(\epsilon) + 2 \sum_{\substack{\eta \in [\mathbf k]\\ \eta(i) = j}} g_\eta(\epsilon).
\end{equation*}

Let $G_\eta := \tilde \Phi(g_\eta)$. The map $\tilde \Phi$ is positive if and only if $G_\eta \geq 0 $ for all $\eta \in [\mathbf k]$. It remains to show that $[G_\eta]_{\eta \in [\mathbf k]}$ is a joint POVM for the $E^{(i)}$ if and only if $\tilde \Phi$ is a positive extension of $\Phi$. By the definition of $\tilde \Phi$ and its unitality, we obtain
\begin{equation*}
 - \frac{2}{k_i} I + 2 \sum_{\substack{\eta \in [\mathbf k]\\ \eta(i) = j}} G_\eta = 2 E_j^{(i)} - \frac{2}{k_i} I \qquad \forall i \in [g], \forall  j \in [k_i -1].
\end{equation*}
Thus, $\tilde \Phi$ is a positive extension of $\Phi$ if and only if the set $\Set{G_\eta}_{\eta \in [\mathbf k]}$ satisfies
\begin{align*}
G_\eta &\geq 0 \qquad \qquad\qquad\qquad\qquad \forall \eta \in [\mathbf k] \\
I &= \sum_{\eta \in [\mathbf k]} G_\eta \\
E_j^{(i)} &= \sum_{\substack{\eta \in [\mathbf k]\\ \eta(i) = j}} G_\eta \qquad \qquad  \forall i \in [g], \forall j \in [k_i -1].
\end{align*}
This is equivalent to $[G_\eta]_{\eta \in [\mathbf k]}$ being a joint POVM for the $\Set{E_1^{(i)} \ldots E_{k_i}^{(i)}}$, since the above conditions also imply
\begin{equation*}
E^{(i)}_{k_i} = I - \sum_{j = 1}^{k_i -1} E^{(i)}_j= \sum_{\eta \in [\mathbf k]} G_\eta - \sum_{j=1}^{k_i-1}\sum_{\substack{\eta \in [\mathbf k]\\ \eta(i) = j}} G_\eta  = \sum_{\substack{\eta \in [\mathbf k]\\ \eta(i) = k_i}} G_\eta\qquad \forall i \in [g].
\end{equation*}

Finally, the third claim follows from the second one, using the standard argument in \cite[Lemma V.2 and Corollary IV.6]{bluhm2018joint}.
\end{proof}

\begin{remark}
From point $(3)$ of Theorem \ref{thm:jewel-compatible-POVM}, it follows in particular that it is enough to check inclusion at level $d$, because the compatibility of a collection of POVMs is not affected by a conjugation with unitaries applied to all POVM elements. This is also a consequence of the well-known fact that for a map $\Phi: \mathcal S \to \mathcal M_d$, where $\mathcal S$ is an operator system, $d$-positivity is equivalent to complete positivity \cite[Theorem 6.1]{Paulsen2002}. As $\mathcal D_{\jewel, \mathbf k}(1)$ is bounded, $\mathcal D_{\jewel, \mathbf k}(d) \subseteq \mathcal D_{E}(d)$ thus holds if and only if $\mathcal D_{\jewel, \mathbf k} \subseteq \mathcal D_{E}$ (see also \cite[Lemma 2.3]{helton2019dilations}, \cite[Corollary 4.6]{bluhm2018joint}).
\end{remark}

\begin{remark}
The fact that $\mathcal D_{\jewel,k}(1)$ is a simplex for all $k \in \mathbb N$ implies that the free spectrahedron $\mathcal D_{A}$ with $\mathcal D_{A}(1) = \mathcal D_{\jewel,k}(1)$ is uniquely defined (in fact, this is true even for matrix convex sets). This follows from \cite[Section 4]{passer2018minimal}. Remark 4.2 of \cite{davidson2016dilations} implies that $\mathcal D_{\jewel,k}(1) \subseteq \mathcal D_E(1)$ if and only if $\mathcal D_{\jewel,k} \subseteq \mathcal D_E$, which by Theorem \ref{thm:jewel-compatible-POVM} implies  that the joint measurability problem is trivial for just a single POVM. As $\mathcal D_{\jewel,\mathbf k}(1)$ is a polytope but no longer a simplex for $\mathbf k \in \mathbb N^{g}$ and $g \geq 2$, the free spectrahedron with this set at level $1$ is no longer unique. Therefore, there exist POVMs which are not jointly measurable.
\end{remark}

The correspondence in the theorem above also extends to the level of balanced compatibility regions / inclusion sets. The theorem below corresponds to \cite[Theorem V.7]{bluhm2018joint} and is a generalization of the latter from binary POVMs to POVMs with an arbitrary number of outcomes.

\begin{thm} \label{thm:delta-is-gamma}
Let $d$, $g \in \mathbb N$ and $\mathbf k \in \mathbb N^{g}$. Then,
\begin{equation*}
\Gamma(g,d,\mathbf{k}) = \Delta(g,d,\mathbf{k}).
\end{equation*}
\end{thm}
\begin{proof}
Let $s \in \mathbb R^g_+$. It holds that $s \in \Gamma(g,d,\mathbf{k})$ if and only if $s_i E^{(i)} + (1-s_i)I_d/k_i$, $i \in [g]$ are jointly measurable for any $d$-dimensional POVMs with $k_i$ outcomes for the $i$-th POVM. Let $\mathcal D_E $ be as in Theorem \ref{thm:jewel-compatible-POVM}. We find that
\begin{align*}
(s_1^{\times (k_1-1)}, \ldots, s_g^{\times (k_g-1)}) \cdot \mathcal D_{\jewel, \mathbf k} \subseteq \mathcal D_E &\iff \mathcal D_{\jewel, \mathbf k} \subseteq \mathcal D_{(2s_1 E^{(1)}-\frac{2s_1}{k_1} I,\ldots,  2s_gE^{(g)}-\frac{2s_g}{k_g} I)} \\
& \iff D_{\jewel, \mathbf k} \subseteq D_{(2F^{(1)}-\frac{2}{k_1} I,\ldots,  2F^{(g)}-\frac{2}{k_g} I)},
\end{align*}
where $F_j^{(i)} = s_i E_j^{(i)} + (1-s_i) \frac{1}{k_i} I_d$ and $j \in [k_i-1]$, $i \in [g]$. Hence, it follows from Theorem \ref{thm:jewel-compatible-POVM} that $s \in \Gamma(g,d,\mathbf{k})$ if and only if the implication 
\begin{equation} \label{eq:Gamma_implication}
\mathcal D_{\jewel, \mathbf k}(1) \subseteq \mathcal D_E(1) \implies \left(s_1^{\times (k_1-1)}, \ldots, s_g^{\times (k_g-1)}\right) \cdot \mathcal D_{\jewel, \mathbf k} \subseteq \mathcal D_E
\end{equation}
is true for all $E = (2 E^{(1)}-\frac{2}{k_1} I,\ldots,  2E^{(g)}-\frac{2}{k_g} I)$. Moreover, $A \mapsto 2 A - (2/k) I$ is a bijective map on $\mathcal M_d^{sa}$ for fixed $k \in \mathbb N$. Thus, any $B_i \in (\mathcal M_{d}^{sa})^{k_i-1}$ can written as $B_i = 2 E^{(i)} - (2/k_i)I$ with $E^{(i)} \in (\mathcal M_{d}^{sa})^{k_i-1}$ for all $i \in [g]$. Therefore, the implication in Equation \eqref{eq:Gamma_implication} holds for fixed $s$ if and only if the implication 
\begin{equation}
\mathcal D_{\jewel, \mathbf k}(1) \subseteq \mathcal D_B(1) \implies \left(s_1^{\times (k_1-1)}, \ldots, s_g^{\times (k_g-1)}\right) \cdot \mathcal D_{\jewel, \mathbf k} \subseteq \mathcal D_B
\end{equation}
is true for all $B \in (\mathcal M_d^{sa})^{\sum_{i=1}^g (k_i-1)}$. Hence, $s \in \Gamma(g,d,\mathbf k)$ is equivalent to $s \in \Delta(g,d,\mathbf{k})$.
\end{proof}

\section{Compatibility results from quantum information theory
} \label{sec:compatibility-from-QIT}

Having established in the previous section the close relation between compatibility and inclusion sets, we next gather results from quantum information theory which provide upper and lower bounds on the sets $\Gamma$. Such bounds translate immediately, via Theorem \ref{thm:delta-is-gamma}, to the corresponding bounds for the sets $\Delta$; we postpone this analysis until Section \ref{sec:discussion}.

\subsection{Upper bounds from MUBs} \label{sec:MUB-bounds}

Mutually unbiased bases (MUBs) yield natural examples of POVMs which are very far from being compatible \cite{wootters1989optimal}. Recall that, in $\mathbb C^d$, a collection of $g$ orthonormal bases $\Set{\psi_j^{(i)}}_{j = 1}^d$, $i \in [g]$, is called \emph{mutually unbiased} if
\begin{equation*}
|\langle \psi_j^{(i)},\psi_u^{(v)}\rangle|^2 = \frac{1}{d}
\end{equation*}
for $i \neq v$ and any $j$, $u \in [d]$. Let $E_j^{(i)} = \psi_j^{(i)}(\psi_j^{(i)})^\ast $ be the corresponding effect operators.
In the case where we construct one MUB from another one by applying a Fourier transform, i.e.\
\begin{equation*}
\psi_k^{(2)} = \frac{1}{\sqrt{d}} \sum_{l = 1}^d e^{2\pi \mathrm{i} \frac{lk}{d}} \psi_{l}^{(1)}
\end{equation*}
we will call these two MUBs \emph{canonically conjugated}.

The maximal number of MUBs in dimension $d$ is $d+1$ and it is known that this bound is attained if $d = p^r$ for a prime number $p$ and $r \in \mathbb N$ \cite{wootters1989optimal}. Apart from that, very few examples are known, see \cite{durt2010mutually} for a review. From \cite{Carmeli2012}, we have the following results on two canonically conjugated MUBs:
\begin{prop}[{\cite[Proposition 5, Example 1 and Proposition 6]{Carmeli2012}}] \label{prop:upper-2-mubs}
Let $E^{(1)}$ and $E^{(2)}$ be the effect operators corresponding to two canonically conjugated MUBs. Then, $ \lambda E^{(1)} + (1-\lambda)I/d$ and $ \mu E^{(2)} + (1-\mu)I/d$ are jointly measurable if and only if
\begin{equation*}
\mu \leq \frac{1}{d}[(d-2)(1-\lambda)+2\sqrt{(1-d)\lambda^2 + (d-2)\lambda + 1}]
\end{equation*}
(equivalently, we can exchange $\lambda$ and $\mu$). Another equivalent form is that the above POVMs are jointly measurable if and only if
\begin{equation*}
\mu + \lambda \leq 1 \qquad \text{or} \qquad \lambda^2 + \mu^2+\frac{2(d-2)}{d}(1-\mu)(1-\lambda) \leq 1.
\end{equation*}
In particular, for $\mu = \lambda$, this simplifies to 
\begin{equation*}
\lambda \leq \frac{1}{2}\left(1+\frac{1}{1+\sqrt{d}}\right).
\end{equation*}
\end{prop}

For more than two MUBs, there is a necessary criterion which generalizes the above in the symmetric case.
\begin{prop}[{\cite[Equation 10]{Designolle2018}}]
Let $\lambda E^{(i)} + (1-\lambda)I/d$, $i \in [g]$ be jointly measurable. Then, it holds that
\begin{equation*}
\lambda \leq \frac{\sqrt{d} + g}{g(\sqrt{d} + 1)}.
\end{equation*}
\end{prop}

\bigskip

There is a different approach to finding necessary conditions for joint measurability developed by H.~Zhu. While it is not restricted to MUBs, it seems to work best for these objects.
We recall Zhu's incompatibility criterion \cite{zhu2015information,zhu2016universal}. Define, for any matrix $A$ with $\tr[A] \neq 0$,
\begin{equation*}
 \overline{\mathcal G}(A):= \frac{\dyad{A^\circ}}{\tr[A]} \in \mathcal M^{sa}_{d^2},
\end{equation*}
where $A^\circ = A - \tr[A] I/d$ and $\Ket{A^\circ}$ is a vectorization of $A^\circ$. For a POVM $E$, we define
\begin{equation*}
 \overline{\mathcal G}(\Set{E_i}_{i \in [k]}) = \sum_{i = 1}^k \overline{\mathcal G}(E_i).
\end{equation*}
\begin{prop}[{\cite[Equations (10,11)]{zhu2015information}}] \label{prop:zhu}
Let $E^{(i)}$, $i \in [g]$ be a collection of compatible POVMs in $\mathcal M_d$. Then,
\begin{equation*}
1 + \min \Set{\tr[H]:H \geq \overline{\mathcal G}(E^{(i)}), \forall i \in [g]} \leq d.
\end{equation*}
\end{prop}
If we are interested in the case of $g$ MUBs, we obtain the following necessary criterion, which appears in \cite{zhu2016universal}. We will provide a proof for convenience.
\begin{prop}
Let $\Set{\psi_j^{(i)}}_{j = 1}^d$, $i \in [g]$ be a collection of MUBs with corresponding POVMs $E^{(i)}$. If $\lambda_i E^{(i)} +(1-\lambda_i)I/d$ are compatible for $\lambda_i \in [0,1]$, then
\begin{equation*}
\sum_{i = 1}^g \lambda_i^2 \leq 1.
\end{equation*}
\end{prop}
\begin{proof}
A straightforward calculation shows that $\tr[(E_j^{(i)})^\circ (E_u^{(v)})^\circ] = 0$ if and only if
\begin{equation*}
\tr[E_j^{(i)} E_u^{(v)}] = \frac{\tr[E_j^{(i)}]\tr[E_u^{(v)}]}{d}.
\end{equation*}
The latter condition is fulfilled by the MUBs for $i \neq v$. Hence, the $\overline{\mathcal G}(E^{(i)})$ are pairwise orthogonal and the same holds for $\overline{\mathcal G}(\tilde{E}^{(i)})$, where $\tilde{E}^{(i)}_j = \lambda_i E^{(i)}_j + (1-\lambda_i) I/d$, $\lambda_i \in [0,1]$. Moreover, both sets of operators are positive semidefinite. Let $P_i$ be the orthogonal projections on the supports of the matrices $\overline{\mathcal G}(\tilde{E}^{(i)})$; it follows that $P:=\sum_{i=1}^g P_i$ is also an orthogonal projection. Consider now a self-adjoint matrix $H$ such that $H \geq \overline{\mathcal G}(\tilde{E}^{(i)})$ for all $i \in [g]$. We have
$$\tr[H] \geq \tr[PHP] = \sum_{i=1}^g \tr[P_i H P_i] \geq \sum_{i=1}^g \tr[P_i \overline{\mathcal G}(\tilde{E}^{(i)}) P_i] = \sum_{i=1}^g \tr[ \overline{\mathcal G}(\tilde{E}^{(i)})].$$
Therefore, by Proposition \ref{prop:zhu} we have
\begin{align*}
d - 1 &\geq \sum_{i = 1}^g \sum_{j = 1}^d \tr[\overline{\mathcal G}(\tilde{E}^{(i)}_j)] \\
&= \sum_{i = 1}^g \sum_{j = 1}^d \lambda_i^2 
 \tr[\overline{\mathcal G}(E^{(i)}_j)] \\
&= d \frac{d-1}{d}\sum_{i = 1}^g \lambda_i^2.
\end{align*}
This proves the claim.
\end{proof}

\subsection{Lower bounds from cloning}

In this section, we will review some results on asymmetric cloning, which will then translate into lower bounds on the inclusion sets for the matrix jewel. See \cite[Section VI]{bluhm2018joint} for a more detailed discussion. Let us define the set of allowed parameters arising from cloning:
\begin{align}
\Gamma^{clone}(g,d):=& \Bigg\{s\in [0,1]^g: \exists \mathcal T: \mathcal M_d^{\otimes g} \to \mathcal M_d \mathrm{~unital~and~completely~positive~linear~map~s.t.\ ~} \label{eq:gamma-clone} \\&  \forall X \in \mathcal M_d, \forall i \in [g], \mathcal T\left(I^{\otimes (i-1)}\otimes X \otimes I^{\otimes (n-i)}\right)=s_iX+(1-s_i)\frac{\tr[X]}{d}I\Bigg\}. \nonumber
\end{align}
A cloning map $\mathcal C$ is a quantum channel from $\mathcal M_d$ to $\mathcal M_d^{\otimes g}$ which maps all pure states $\sigma$ as close as possible to $\sigma^{\otimes g}$. Often, the worst case single copy fidelity $F_i$ is used to quantify the error with respect to a perfect cloning device (which is impossible to implement). Here,
\begin{equation*}
F_i(\mathcal C) := \inf_{\psi \in \mathcal S(\mathcal H)\mathrm{~pure}} \tr[\mathcal C(\psi)I^{\otimes (i-1)}\otimes \psi \otimes I^{\otimes (g-i)}], \qquad \forall i \in [g].
\end{equation*}

The following proposition clarifies the connection between asymmetric cloning and our definition of $\Gamma^{clone}(g,d)$, by showing that, without any loss in single copy fidelities, any cloning map can be assumed to have depolarizing marginals. It uses ideas which can be found in  \cite{Werner1998} (see also \cite{hashagen2016universal}) and will allow us to identify the $\mathcal T$ in Equation \eqref{eq:gamma-clone} with the dual of the cloning map.

\begin{prop} Let $\mathcal C: \mathcal M_d \to \mathcal M_d^{\otimes g}$ be a quantum channel with $F_i(\mathcal C) = \eta_i$ $\forall i \in [g]$. Then, there is a channel $\tilde{\mathcal C}: \mathcal M_d \to \mathcal M_d^{\otimes g}$ such that 
\begin{equation*}
\tr[\tilde{\mathcal C}(\psi)I^{\otimes (i-1)}\otimes \psi \otimes I^{\otimes (n-i)}] = \nu_i \geq \eta_i \qquad \forall \psi \in \mathcal S(\mathbb C^{d}) \mathrm{~pure}, \forall i \in [g].
\end{equation*}
Moreover, $\tilde{\mathcal C}$ can be chosen such that
\begin{equation*}
\tilde{\mathcal C}_i(A) = \tr_{i^{c}}[\tilde{\mathcal C}(A)] = \lambda_i A + (1 - \lambda_i)\frac{\tr[A]}{d}I_d \qquad \forall A \in \mathcal M_d.
\end{equation*}
Here, $\lambda_i = (d \nu_i -1)/(d-1) \in [0,1]$ $\forall i \in [g]$ and $\tr_{i^{c}}[\cdot]$ denotes the partial trace over all systems but the $i$-th one.
\end{prop}

\begin{proof}
We claim that we can choose $\tilde{C}$ as a symmetrized version of $\mathcal C$, i.e.\
\begin{equation*}
\tilde{\mathcal C}(A) = \int_{\mathcal U(d)} (U^{\otimes g}) \mathcal C(U^\ast A U) (U^{\otimes g})^\ast d\mu(U) \qquad A \in \mathcal M_d.
\end{equation*}
Here, $\mu$ is the normalized Haar measure on the unitary group. The marginals of this map are
\begin{align*}
\tilde{\mathcal C}_i(A) &= \int_{\mathcal U(d)}\tr_{i^c}[(U^{\otimes g}) \mathcal C(U^\ast A U) (U^{\otimes g})^\ast] d\mu(U) \\
&= \int_{\mathcal U(d)} U \mathcal C_i(U^\ast A U) U^\ast d\mu(U) \qquad \forall A \in \mathcal M_d,
\end{align*}
where we have written $\mathcal C_i(A) := \tr_{i^c}[\mathcal C(A)]$.
We observe furthermore that for any $V \in \mathcal U(d)$ and $A \in \mathcal M_d$,
\begin{align*}
V \tilde{\mathcal C}_i(A) V^\ast &= V \int_{\mathcal U(d)} U \mathcal C_i(U^\ast A U) U^\ast  d\mu(U)V^\ast \\
&=   \int_{\mathcal U(d)} W \mathcal C_i(W^\ast V A V^\ast W) W^\ast  d\mu(W)\\
&= \tilde{\mathcal C}_i(V A V^\ast),
\end{align*} 
where we have used left-invariance of the Haar measure in the second line. Thus,
\begin{equation}\label{eq:covariance}
V \tilde{\mathcal C}_i(\cdot) V^\ast = \tilde{\mathcal C}_i(V \cdot V^\ast).
\end{equation}
Let us compute the single copy fidelities.
\begin{align*}
\tr[\tilde{\mathcal C}(\psi)I^{\otimes (i-1)}\otimes \psi \otimes I^{\otimes (n-i)}] &= \int_{\mathcal U(d)} \tr[\mathcal C(U^\ast\psi  U)I^{\otimes (i-1)}\otimes U^\ast\psi U \otimes I^{\otimes (n-i)}] d\mu(U)\\
&\geq \int_{\mathcal U(d)} \eta_i d\mu(U) = \eta_i.
\end{align*}
Here, we have used that $U^\ast \psi U$ is a pure state and that the Haar measure is positive. This shows the first assertion. 
Let us now prove the third assertion. Let $\tau_{0i}$ be the Choi matrix of $\tilde{\mathcal C}_i$, i.e.\ $\tau_{0i} = (\mathrm{Id}_d \otimes \tilde{\mathcal C}_i)(\Omega)$, where $\Omega$ is  the maximally entangled state
\begin{equation*}
\Omega := \frac{1}{d} \sum_{i,j = 1}^d (e_i \otimes e_i)(e_j \otimes e_j)^\ast
\end{equation*}
and $\Set{e_i}_{i = 1}^d$ is an orthonormal basis of $\mathbb C^d$. Let $V \in \mathcal U(d)$. Then,
\begin{align*}
(\overline V \otimes V) \tau_{0i} (\overline V \otimes V)^\ast &= (\mathrm{Id}_d \otimes \tilde{\mathcal C}_i)((\overline V \otimes V)\Omega(\overline V \otimes V)^\ast)\\
&=(\mathrm{Id}_d \otimes \tilde{\mathcal C}_i)(\Omega) =\tau_{0i}
\end{align*}
where we have used Equation \eqref{eq:covariance} and the well-known trick $(A \otimes I_d)\Omega = (I_d \otimes A^T)\Omega$ for any $A \in \mathcal M_d$. The above invariance implies that $\tau_{0i}$ is an isotropic state and is therefore of the form \cite[Section 3.1.3]{Keyl2002}
\begin{equation*}
\tau_{0i}=(1 -\lambda_i)\frac{1}{d^2}I_{d^2} + \lambda_i \Omega, \qquad \lambda_i \in \left[-\frac{1}{d^2-1},1\right].
\end{equation*}
By the Choi-Jamio{\l}kowski isomorphism, this is equivalent to 
\begin{equation*}
\tilde{\mathcal C}_i(A) = \lambda_i A + (1-\lambda_i)\frac{\tr[A]}{d} I_d \qquad \forall A \in \mathcal M_d.
\end{equation*}
With this expression, we can explicitly compute the single copy fidelities
\begin{align*}
\tr[\tilde{\mathcal C}(\psi)I^{\otimes (i-1)}\otimes \psi \otimes I^{\otimes (n-i)}] &= \tr[\tilde{\mathcal C}_i(\psi  )\psi ]
= \lambda_i + \frac{1 -\lambda_i}{d}.
\end{align*}
This proves the second assertion as well as the expression for $\lambda_i$ in terms of $\nu_i$.
\end{proof}

Therefore, we can now use $\mathcal T = \tilde{\mathcal C}^\ast$ in Equation \eqref{eq:gamma-clone}, which shows that $\Gamma^{clone}(g,d)$ indeed arises from optimal asymmetric cloning. The exact form of $\Gamma^{clone}(g,d)$ has been computed in \cite{kay2016optimal,studzinski2014group}, using different methods. To obtain the theorem below from \cite{kay2016optimal}, one needs to perform the necessary transform from $\nu_i$ to $\lambda_i$.
\begin{thm}[{\cite[Theorem 1, Section 2.3]{kay2016optimal}}]
For any $g$,$d \geq 2$
\begin{align*}
\Gamma^{clone}(g,d) =& \left\{s \in [0,1]^g:(g+d-1)\left[g-d^2+d+(d^2-1)\sum_{i=1}^g s_i\right]\right. \\ &\left.\leq \left(\sum_{i=1}^g \sqrt{s_i(d^2-1)+1}\right)^2\right\}.
\end{align*}
In particular, for $s_1 = \ldots =s_g$, the maximal value is
\begin{equation*}
s_{max} = \frac{g+d}{g(d+1)}.
\end{equation*}
\end{thm}
In the symmetric case, the optimal cloning map is unique \cite{Werner1998, Keyl2002}.
The following proposition shows that cloning gives indeed a lower bound on the balanced compatibility region.

\begin{prop} \label{prop:cloning-lower}
Let $g$, $d \in \mathbb N$ and $\mathbf k \in \mathbb N^g$ and $k_{max} = \max_{i \in [g]}k_i$. Then, it holds that 
\begin{equation*}
\Gamma^{clone}(g,k_{max} d) \subseteq \Gamma(g,d,\mathbf{k}).
\end{equation*}
\end{prop}
\begin{proof}
Using 
\begin{equation*}
G_{j_1, \ldots, j_g}=\mathcal T\left(E^{(1)}_{j_1} \otimes \ldots \otimes E^{(g)}_{j_g}\right)
\end{equation*}
as a joint POVM, where $\mathcal T$ is the map from Equation \eqref{eq:gamma-clone}, it is clear that 
\begin{equation} \label{eq:clone-lin-inclusion}
\Gamma^{clone}(g,D) \subseteq \Gamma^{lin}(g,D,\mathbf{k}), \qquad \forall \mathbf k \in \mathbb N^{g}, \forall D \in \mathbb N.
\end{equation}
The assertion follows then by Proposition \ref{prop:lin-gamma-inclusion}.
\end{proof}
\begin{remark}
Note that the left hand side of Equation \eqref{eq:clone-lin-inclusion} is independent of $\mathbf k$, since the cloning map is designed to clone states, not measurements, such that we can perform any kind of measurement on the approximate clones.
\end{remark}

\section{Lower bounds from symmetrization} \label{sec:symmetrization}

In this section, we give lower bounds on $\Delta(g,d,\mathbf{k})$ by considering its inclusion inside more symmetric spectrahedra. We start with a single point.

\begin{thm}\label{thm:symmetric-jewel}
Let $g$, $d \in \mathbb N$, $k_j \in \mathbb N$, $\mathbf{k} = (k_1, \ldots, k_g)$. Then, 
\begin{equation*}
\frac{1}{2d} \left( \frac{1}{k_1-1}, \ldots, \frac{1}{k_g-1}\right) \in \Delta(g,d,\mathbf k).
\end{equation*}
\end{thm}
\begin{proof}
We consider a symmetrization of the matrix jewel, which we denote as
\begin{equation*}
\mathcal D_{S\jewel, \mathbf{k}} := \mathcal W_{max}(\mathrm{conv}\{-\mathcal D_{\jewel, \mathbf{k}}(1)\cup \mathcal D_{\jewel, \mathbf{k}}(1)\}).
\end{equation*}
Since the matrix jewel is a polytope on the first level, $\mathcal D_{S\jewel, \mathbf{k}}$ is indeed a free spectrahedron. It holds that 
\begin{equation*}
\mathcal D_{\jewel, \mathbf{k}} \subseteq \mathcal D_{S\jewel, \mathbf{k}},
\end{equation*}
since the inclusion holds at level $1$ and $\mathcal D_{S\jewel, \mathbf{k}}$ is a maximal spectrahedron (see also \cite[Remark 4.2]{davidson2016dilations}). Let $\lambda \in [0,1]^g$ be such that
\begin{equation*}
\lambda \cdot \mathcal D_{S\jewel, \mathbf{k}}(1) \subseteq \mathcal D_{\jewel, \mathbf{k}}(1).
\end{equation*}
Then, for any $B \in (\mathcal M_d^{sa})^{\sum_{j=1}^g(k_j-1)}$, the implication 
\begin{equation*}
\mathcal D_{\jewel, \mathbf{k}}(1) \subseteq \mathcal D_B(1) \implies \frac{1}{2d} \lambda \cdot \mathcal D_{S\jewel, \mathbf{k}} \subseteq \mathcal D_B
\end{equation*}
holds by \cite[Proposition VII.2]{bluhm2018joint}, which generalizes \cite[Theorem 1.4]{helton2019dilations} to the complex setting.  We can apply this result to the asymmetrically scaled spectrahedron, since $\lambda \cdot \mathcal D_{A}(n) \subseteq \mathcal D_B(n)$ if and only if $\mathcal D_{A}(n) \subseteq \mathcal D_{\lambda \cdot B}(n)$ for any free spectrahedra $\mathcal D_{A}$, $\mathcal D_{B}$ and any $n \in \mathbb N$. Therefore, $\lambda/(2d) 
\in \Delta(g,d,\mathbf k)$. We only need to find the largest valid $\lambda$. As can be seen from comparing the extreme points, the symmetrization carries through the direct sum construction of the matrix jewel,
\begin{equation*}
\mathcal D_{S\jewel, \mathbf{k}}(1) = \bigoplus_{i = 1}^g \mathrm{conv}\{-\mathcal D_{\jewel, k_i}(1) \cup \mathcal D_{\jewel, k_i}(1)\}.
\end{equation*}
 We note that $X \in \mathcal D_A$ if and only if $X \in \mathcal D_{A \otimes I}$, which are the elements appearing as summands in the direct sum of free spectrahedra. By Lemma \ref{lem:sumfallsapart},
the conditions on $\lambda$ reduce to
\begin{equation} \label{eq:largest_lambda}
\lambda_i \mathrm{conv}\{-\mathcal D_{\jewel, k_i}(1) \cup \mathcal D_{\jewel, k_i}(1)\} \subseteq  \mathcal D_{\jewel, k_i}(1)
\end{equation}
for each $i \in [k]$. We recall that $\mathcal D_{\jewel, k_i}(1)$ has extreme points $-(k_i/2) e_j$, $j \in [k_i -1]$ and $(k_i/2)(1, \ldots, 1)$ by Lemma \ref{lem:extremal-points-jewel-base}. We can write $-(k_i/2)(1, \ldots, 1)$ and $(k_i/2) e_j$ as a convex combination of extreme points of $(k_i -1) \mathcal D_{\jewel, k_i}(1)$, as
\begin{equation*}
-\frac{1}{k_i-1}\frac{k_i}{2}(1, \ldots, 1) = \frac{1}{k_i-1}\sum_{j = 1}^{k_i-1} \left(-\frac{k_i}{2} e_j\right)
\end{equation*}
and
\begin{equation} \label{eq:explicit_convex_combination}
\frac{1}{k_i-1}\frac{k_i}{2}e_j = \frac{1}{k_i-1}\frac{k_i}{2}(1, \ldots, 1)+\frac{1}{k_i-1}\sum_{\substack{l = 1\\l \neq j}}^{k_i-1} \left(-\frac{k_i}{2} e_l\right).
\end{equation}
Therefore, $-(k_i/2)(1, \ldots, 1)$ and $(k_i/2) e_j \in (k_i -1) \mathcal D_{\jewel, k_i}(1)$ for all $j \in [k_i -1]$. Thus, $\lambda_i = 1/(k_i -1)$ is a valid choice in Equation \eqref{eq:largest_lambda}.
\end{proof}
Furthermore, we can approximate the matrix jewel by sets for which we know the inclusion constants. A convenient choice for such a set is the matrix diamond. A similar idea has been used in \cite[Section 2]{Passer2018a}. To state the result, we write
\begin{equation*}
\mathrm{QC}_g := \Set{s \in [0,1]^g: \sum_{i = 1}^g s_i^2 \leq 1}
\end{equation*}
for the positive part of the Euclidean unit ball.
\begin{thm}\label{thm:jewel-vs-diamond}
Let $g$, $d \in \mathbb N$, $k_j \in \mathbb N$, $\mathbf{k} = (k_1, \ldots, k_g)$. Then, 
\begin{equation*}
\left(\frac{1}{(k_1 - 1)^2}, \ldots, \frac{1}{(k_g - 1)^2}\right) \cdot \Delta\left(g,d,2^{\times \sum_{i = 1}^g (k_i - 1)}\right)\subseteq\Delta(g,d,\mathbf k).
\end{equation*}
In particular, 
\begin{equation*}
\left(\frac{1}{(k_1 - 1)^2}, \ldots, \frac{1}{(k_g - 1)^2}\right) \cdot \mathrm{QC}_{\sum_{i = 1}^g (k_i - 1)}\subseteq \Delta(g,d,\mathbf k).
\end{equation*}
\end{thm}
\begin{proof}
We observe that
\begin{equation*}
\mathcal D_{\jewel, k_i}(1) \subseteq \frac{k_i(k_i-1)}{2} \cdot \mathcal D_{\diamond, k_i -1}(1).
\end{equation*}
This follows from the computation of the $\ell_1$ norms of the extremal points of the jewel base found in Lemma \ref{lem:extremal-points-jewel-base}. Moreover, the matrix diamond is the maximal spectrahedron for the $\ell_1$-ball. Thus, together with Lemma \ref{lem:sumfallsapart},
\begin{equation*}
\mathcal D_{\jewel, \mathbf{k}} \subseteq  \left(\frac{k_1(k_1-1)}{2}, \ldots, \frac{k_g(k_g-1)}{2}\right) \cdot \mathcal D_{\diamond, \sum_{i = 1}^g(k_i -1)}
\end{equation*}
(see again \cite[Remark 4.2]{davidson2016dilations}). Furthermore, we need to find the largest $\lambda_i \geq 0$ for $i \in [g]$ such that 
\begin{equation*}
\lambda_i \cdot \mathcal D_{\diamond, k_i -1}(1) \subseteq \mathcal D_{\jewel, k_i}(1).
\end{equation*}
The extreme point of the matrix diamond are $\pm e_j$ for $j \in [k_i - 1]$. It holds that $\pm \lambda_i e_j \subseteq \mathcal D_{\jewel, k_i}(1)$ if and only if
\begin{equation*}
\pm \lambda_i \in \left[-\frac{k_i}{2}, \frac{1}{k_i - 1} \frac{k_i}{2} \right].
\end{equation*}
This follows directly from 
Lemma \ref{lem:extremal-points-jewel-base} and Equation \eqref{eq:explicit_convex_combination}. Thus, $\lambda_i \leq k_i/(2(k_i-1))$. 
From Lemma \ref{lem:sumfallsapart}, we infer that
\begin{equation*}
\left(\frac{k_1}{2(k_1 - 1)} , \ldots, \frac{k_g}{2(k_g - 1)} \right) \cdot \mathcal D_{\diamond, \sum_{j = 1}^g(k_j -1)}(1) \subseteq \mathcal D_{\jewel, \mathbf{k}}(1).
\end{equation*}
Let $B \in (\mathcal M_d^{sa})^{\sum_{j=1}^g(k_j-1)}$. Now, by the previous reasoning, the implication 
\begin{align*}
&\left(\frac{k_1}{2(k_1 - 1)} , \ldots, \frac{k_g}{2(k_g - 1)} \right) \cdot \mathcal D_{\diamond, \sum_{j = 1}^g(k_j -1)}(1) \subseteq \mathcal D_{\jewel, \mathbf{k}}(1)\subseteq \mathcal D_B(1) \implies\\ &  \left(s_1 \frac{1}{(k_1-1)^2}, \ldots, s_g\frac{1}{(k_g-1)^2}  \right) \cdot \mathcal D_{\jewel, \mathbf{k}} \subseteq \left(s_1\frac{k_1}{2(k_1 - 1)} , \ldots, s_g\frac{k_g}{2(k_g - 1)} \right) \cdot \mathcal D_{\diamond, \sum_{j = 1}^g(k_j -1)} \subseteq \mathcal D_B
\end{align*}
holds for all $s \in \Delta(g,d,2^{\times \sum_{j = 1}^g(k_j -1)})$. As $B$ was arbitrary, this proves the first assertion. The second follows from \cite[Theorem VII.7]{bluhm2018joint}, which adapts results from \cite{passer2018minimal}.
\end{proof}

\section{Incompatibility witnesses and the matrix cube}\label{sec:incompatibility-witnesses-binary}

In this section we introduce the notion of \emph{incompatibility witnesses} in the case of tuples of binary POVMs. The case of general POVMs will be treated in the next section. The terminology is borrowed from the theory of entanglement where entanglement witnesses allow one to detect entanglement in quantum states. In the same way, elements from the matrix diamond detect the incompatibility of POVMs as will be shown below. We would like to point out that 
a related notion was recently introduced by A.~Jen{\v{c}}ov{\'a} in \cite{jencova2018incompatible}; see also \cite{carmeli2018quantum} for yet another notion of incompatibility witness.

The use of incompatibility witnesses is twofold. On the one hand, they can be used to certify incompatibility of a given set of POVMs in the regime where this becomes a hard computational problem. In this respect, they play the same role as entanglement witnesses. On the other hand, they can be used to prove new bounds on the compatibility region. To illustrate this, we recover in Proposition \ref{prop:planar-qubits} a result originally obtained in \cite{uola2016adaptive}.

Let us start with a simple calculation motivating the new definition. Recall from \cite[Theorem V.3]{bluhm2018joint} (or from Theorem \ref{thm:jewel-compatible-POVM}) that $g$ quantum effects $E_1, \ldots, E_g \in \mathcal M_d$ are compatible if and only if for all elements of the matrix diamond $X \in \mathcal D_{\diamondsuit,g}$, it holds that
\begin{equation}\label{eq:condition-2E-I-X}
\sum_{i=1}^g (2E_i - I_d) \otimes X_i \leq I_{dn}.
\end{equation}
Recall that for a $g$-tuple $(X_1, \ldots, X_g) \in \mathcal M_n^{sa}$ to be an element of the matrix diamond, it needs to satisfy the following conditions:
$$\forall \epsilon \in \{\pm 1\}^g, \qquad \sum_{i=1}^g \epsilon_i X_i \leq I_n.$$
Let us now show, by a simple and direct computation, why compatible effects $E_1, \ldots, E_g$ must satisfy condition \eqref{eq:condition-2E-I-X}, for any choice of $X$ as above. We write $G$ for the joint POVM associated with $E_1, \ldots, E_g$. Then, 
\begin{align*}
\sum_{i=1}^g (2E_i - I_d) \otimes X_i &= \sum_{i=1}^g \left[ \sum_{\substack{\eta \in \{0,1\}^g\\\eta_i = 0}} G_\eta-\sum_{\substack{\eta \in \{0,1\}^g\\\eta_i = 1}} G_\eta \right]\otimes X_i\\
&= \sum_{i=1}^g  \sum_{\eta \in \{0,1\}^g} (-1)^{\eta_i} G_\eta\otimes X_i\\
&= \sum_{\eta \in \{0,1\}^g} G_\eta \otimes \left[ \sum_{i=1}^g  (-1)^{\eta_i} X_i\right]\\
&\leq \sum_{\eta \in \{0,1\}^g} G_\eta \otimes I_n\\
&= I_{dn}.
\end{align*}

The computation above justifies the following definition. 

\begin{defi}\label{def:incompatibility-witness-binary}
A $g$-tuple of self-adjoint matrices $X \in (\mathcal M_n^{sa})^g$ is called an \emph{incompatibility witness} if one of the following equivalent conditions holds: 
\begin{enumerate}
	\item $X$ is an element of the matrix diamond $\mathcal D_{\diamondsuit,g}$
	\item for all sign vectors $\epsilon \in \{\pm 1\}^g$, $\sum_{i=1}^g \epsilon_i X_i \leq I_n$
	\item for all sign vectors $\epsilon \in \{\pm 1\}^g$, $\|\sum_{i=1}^g \epsilon_i X_i \|_\infty \leq 1$.	
\end{enumerate}
\end{defi}

We can now restate the second claims in \cite[Theorem V.3]{bluhm2018joint} and Theorem \ref{thm:jewel-compatible-POVM} (applied to binary POVMs) as follows. 

\begin{prop}
A set of $d$-dimensional quantum effects $(E_1, \ldots, E_g)$ is jointly measurable if and only if, for any incompatibility witness $X$, condition \eqref{eq:condition-2E-I-X} holds. Moreover, one can restrict the size of the incompatibility witness to be $d$. 
\end{prop}

Deciding whether a $g$-tuple of operators is an incompatibility witness requires to check $2^g$ matrix inequalities of size $n$, a task which is computationally intractable for large $g$ although it can be formulated as a semidefinite program. We relate this question to another free spectrahedral inclusion problem, that of the \emph{complex matrix cube}. Recall from \cite{helton2019dilations} that the matrix cube is the free spectrahedron 
\begin{align}
\mathcal D_{\square, g} &:= \bigsqcup_{n=1}^\infty \left\{ X \in (\mathcal M_n^{sa})^g \, : \, \|X_i\|_\infty \leq 1, \, \forall i \in [g]\right\} \label{eq:matrix-cube}\\
&=  \bigsqcup_{n=1}^\infty \left\{ X \in (\mathcal M_n^{sa})^g \, : \, \sum_{i=1}^g c_i \otimes X_i \leq I_{2gn}\right\}, \nonumber
\end{align}
where the vectors $c_1, \ldots c_g \in \mathbb C^{2g}$ are given by
$$c_i = (e_i, -e_i).$$
We have the following result. 
\begin{prop}\label{prop:incompatibility-witness-vs-cube}
A $g$-tuple $X \in (\mathcal M_d^{sa})^g$ is an incompatibility witness if and only if $\mathcal D_{\square, g}(1) \subseteq \mathcal D_X(1)$. Moreover, we have
\begin{equation}\label{eq:matrix-cube-relaxation}
\mathcal D_{\square, g} \subseteq \mathcal D_X \implies \mathcal D_{\square, g}(1) \subseteq \mathcal D_X(1) \implies \vartheta^{\mathbb C}_{g,d}\mathcal D_{\square, g} \subseteq \mathcal D_X.
\end{equation}
Here, the $\vartheta_{g,d}^{\mathbb C}$ are the symmetric inclusion constants for the \emph{complex} matrix cube, i.e.~the $s \in \Delta_{\mathcal D_{\square, g}}(g,d,2^{\times g})$ for which $s_1 = \ldots = s_g = \vartheta_{g,d}^{\mathbb C}$.
\end{prop}
\begin{proof}
The convex set inclusion $\mathcal D_{\square, g}(1) \subseteq \mathcal D_X(1)$ can be checked at the level of extremal points of the cube $\mathcal D_{\square, g}(1)$, which are the $2^g$ sign vectors $\epsilon \in \{\pm 1\}^g$. The resulting conditions are precisely the ones from Definition \ref{def:incompatibility-witness-binary}. Equation \eqref{eq:matrix-cube-relaxation} follows from the definition of the inclusion constants.
\end{proof}

\begin{remark}
The inclusion constants $\vartheta_{g,d}^{\mathbb C}$ above are the maximal elements $s \in \Delta_{\mathcal D_{\square,g}}(g,d,2^{\times g})$ such that $s_1 = \ldots = s_g$. They are known to possess a dimension independent lower bound, $g^{-1/2} \leq \vartheta_{g,d}^{\mathbb C}$ \cite[Section 6]{passer2018minimal}, which is known to be tight for $d$ large enough.
\end{remark}

\begin{remark}
The chain of implications \eqref{eq:matrix-cube-relaxation} suggests an efficient numerical procedure to determine, up to some precision, whether a given $g$-tuple of self-adjoint operators is an incompatibility witness. This is because the first and the last free spectrahedral inclusions can be formulated as an SDP, as follows:
\begin{align*}
	\text{maximize} \quad &s\\
	\text{subject to} \quad &\exists \Phi:\mathbb C^{2g} \to \mathcal M_d \text{ unital, completely positive}\\
	&s\Phi(c_i) = X_i \quad \forall i \in [g].
\end{align*}
If the value $s^*$ of the SDP above is such that $s^* \geq 1$, we conclude that the first inclusion in \eqref{eq:matrix-cube-relaxation} holds, so $X$ is an incompatibility witness. On the other hand, if the optimal value is such that $s^*<\vartheta^{\mathbb C}_{g,d}$, we conclude that $X$ is \emph{not} an incompatibility witness. However, if $s^* \in [\vartheta^{\mathbb C}_{g,d}, 1)$, we cannot conclude anything. Finally, let us point out that the SDP above has $3g+1$ constraints of size $d$, hence it is more tractable than the original brute-force condition, requiring $2^g$ matrix inequalities. 
\end{remark}

We end this section with an example of an application of the theory of incompatibility witnesses. We shall prove that the upper bound derived in \cite{uola2016adaptive} for the amount of noise needed to make a $g$-tuple of ``planar'' qubit POVMs jointly measurable can also be understood in the framework of incompatibility witnesses.

Recall that a \emph{planar qubit POVM} is a binary qubit POVM with effects which depend on only two Pauli operators (we choose $\sigma_X$ and $\sigma_Y$ below). We use the standard Pauli matrices
\begin{equation*}
\sigma_X = \begin{pmatrix}
0  & 1 \\ 1& 0
\end{pmatrix}, \qquad
\sigma_Y = \begin{pmatrix}
0  & -i \\ i& 0
\end{pmatrix}, \qquad
\sigma_Z = \begin{pmatrix}
1  & 0 \\ 0&-1
\end{pmatrix}.
\end{equation*}

In the case of planar qubit POVMs defined by vectors in the complex plane with angles in arithmetic progression, we have the following result.  

\begin{lem}
Let $X = (X_1, \ldots, X_g)$ where $X_j$ are planar qubit observables
\begin{equation}\label{eq:planar-qubit-observables}
X_j = \cos(j\pi/g) \sigma_X + \sin(j \pi/g) \sigma_Y, \qquad j \in [g].
\end{equation}
Then, $\lambda X$ is a incompatibility witness if and only if $|\lambda| \leq \sin(\pi/(2g))$. 
\end{lem}
\begin{proof}
Let $\epsilon \in \{\pm 1\}^g$. The condition $\|\sum_j \epsilon_j \lambda X_j \|_\infty \leq 1$ reduces in this case, using the Bloch ball picture, to 
$$|\lambda|\left\| \left( \sum_{j=1}^g \epsilon_j \cos(j\pi/g), \sum_{j=1}^g \epsilon_j \sin(j\pi/g) \right) \right\|_2 \leq 1 \iff |\lambda| \left|\sum_{j=1}^g \epsilon_j \omega^j \right| \leq 1,$$
where $\omega = \exp(2\pi i /(2g))$ is a $2g$-th root of unity. Note that choosing $\epsilon \equiv 1$ gives
$$1 \geq |\lambda| \left|\sum_{j=1}^g  \omega^j \right| = \frac{2|\lambda|}{|1-\omega|} = \frac{|\lambda|}{\sin(\pi/(2g))},$$
proving one direction of the conclusion. For the other direction, note that $-\omega^j = \omega^{g+j}$, hence the signed sum of roots of unity corresponds to a sum of a subset of size $g$ of $2g$-roots of unity. The conclusion will follow from the following claim, proving that any optimizer must be a rotation of the $\epsilon \equiv 1$ case. 

\medskip

\textbf{Claim.} The maximization problem 
$$\max_{J \subseteq [2g]} \left| \sum_{j \in J} \omega^j \right| $$
is attained for a subset $J_0$ with cardinality $g$ and such that the set $\{\omega^j\}_{j \in J_0}$ is contained in some half-plane of $\mathbb C = \mathbb R^2$. Furthermore, if $j \in J_0$, then $j + g \not \in J_0$ (the sums are considered modulo $2g$). 

\medskip

Indeed, let $J$ be any maximizer, and let $s:= \sum_{j \in J} \omega^j$. We show that $\{\omega^j\}_{j \in J}$ lies in the half-plane $\{z \in \mathbb R^2 \, : \, \langle s, z \rangle \geq 0\}$. We will use the following fact: For two non-zero vectors $a$, $b$ in a real Hilbert space, $\langle a, b \rangle \geq 0$ implies $|a+b|>|a|$.
Assume that there is some $j \in J$ with
$\langle s, \omega^j \rangle <0$. 
If $g+j \notin J$, replacing $j$ with $g+j$ (taken cyclically) would increase the modulus of the sum, contradicting maximality. This is true, because the sum $s^\prime$ after replacement can be written $s^\prime = s - 2\omega^j$ and $\langle s, -\omega^{j}\rangle > 0$. Then $|s^\prime| > s$ by the fact above.
If $g+j \in J$, the two contributions cancel, and we can consider $J' = J \setminus \{j,g+j\}$ and iterate. So, there is no $j \in J$ such that
$\langle s, \omega^j \rangle <0$. Conversely, if $j \in [2g]$ such that $\langle s, \omega^j \rangle \geq 0$, then $j \in J$. If this was not the case, we would have $|s + \omega^j| > |s|$ which contradicts maximality. Hence, $|J| \geq g$. Assume $|J| > g$. Then, there is an $l \in J$ such that also $g+l \in J$. By the above, this implies $- \langle s, \omega^l \rangle \geq 0$ and thus $|s - \omega^{l}| > |s|$. Removing $l$ from $J$ would thus increase the modulus, contradicting maximality.

The above claim implies that $\sum_{j \in J_0} \omega^j = \omega^{k} \sum_{j \in [g]} \omega_j$ for some $k \in [2g]$, as there are no more than $g+1$ consecutive $\omega^{j}$ in a half-space.
This proves the assertion since $|\omega^k|=1$. 
\end{proof}

\begin{prop} \label{prop:planar-qubits}
Let $g$ be a fixed positive integer, and consider the quantum effects
$$E_j = \frac 1 2 (I_2 + t_j X_j), \qquad j \in [g],$$
for some $t_j \in [0,1]$, where $X_j$ have been defined in \eqref{eq:planar-qubit-observables}. If the above effects are jointly measurable, then 
$$\sum_{j=1}^g t_j \leq \frac{1}{\sin(\pi/(2g))}.$$
\end{prop}
\begin{proof}
From the previous lemma, we know that $\sin(\pi/(2g))X$ is an incompatibility witness, hence so is $\sin(\pi/(2g))X^{T}$
$$\sin(\pi/(2g))\sum_{j=1}^g t_j X_j \otimes X_j^{T} \leq I_4.$$
Let $\Omega = 1/2 \sum_{i,j =1}^2 (e_i \otimes e_i)(e_j \otimes e_j)^\ast$ be the maximally entangled state, where $\Set{e_1,e_2}$ is the basis of $\mathbb C^2$ with respect to which we transpose. By taking the Hilbert-Schmidt inner product of the previous inequality with $\Omega$, we obtain
$$\sin(\pi/(2g))\sum_{j=1}^g t_j \leq 1,$$
proving the claim. Here, we have used $\tr[\Omega A \otimes B] = 1/2\tr[B^TA]$ and $\tr[\sigma_X \sigma_Y] = 0$, $\sigma_X^2 = I_2 = \sigma_Y^2$, by which $\tr[X_j^2] = 2$ $\forall j \in [g]$.
\end{proof}
\begin{cor} \label{cor:planar-qubits}
The proposition above implies the following upper bound for the balanced compatibility regions $\Gamma$ introduced in \cite{bluhm2018joint} for binary POVMs: for all $g \geq 2$, 
$$\Gamma(g,2,2^{\times g}) \subseteq \left\{s \in [0,1]^g \, : \, \sum_{j=1}^g s_j \leq\frac{1}{\sin(\pi/(2g))} \right\}.$$ 
\end{cor}

\begin{remark}
Very similar ideas were used in the proof of \cite[Theorem VIII.8]{bluhm2018joint}. There, it was shown that if $F_1, \ldots, F_g$ are anti-commuting, self-adjoint, unitary $d \times d$ matrices, then the $g$-tuple $(s_1 F_1, \ldots, s_g F_g)$ is an incompatibility witness for any unit norm vector $s$ (see also \cite{kunjwal2014quantum} for a different use of the same matrices in quantum theory). As above, this observation, together with the ``maximally entangled state trick'' yields upper bounds on the sets $\Gamma(g,d, 2^{\times g})$.
\end{remark}

\section{Incompatibility witnesses -- the general case}\label{sec:incompatibility-witnesses-general}

We generalize here the notion of incompatibility witnesses introduced in the previous section for binary POVMs to the case of POVMs with arbitrary number of outcomes. As in previous Sections, we will identify vectors in $\mathbb C^k$ with diagonal $k \times k$ matrices.

\begin{defi}
Given a $g$-tuple of positive integers $\mathbf{k}$, we call the elements of the matrix jewel $\mathcal D_{\jewel, \mathbf k}$ \emph{incompatibility witnesses}. An incompatibility witness $X \in \mathcal D_{\jewel, \mathbf k}(n)$ has the property that for all compatible POVMs $E^{(1)}, \ldots, E^{(g)}$ having $k_i$ outcomes, respectively, the following inequality is satisfied
$$\sum_{i=1}^g \sum_{j=1}^{k_i-1} \left( 2 E^{(i)}_j - \frac{2}{k_i}I \right) \otimes X_{ij} \leq I_{dn}.$$ 
\end{defi}

In order to decide whether a given $\sum_{i=1}^g (k_i-1)$-tuple $X$ is an incompatibility witness, one has to check $\prod_{i=1}^g (k_i-1)$ matrix inequalities (see Definition \ref{def:jewel} and Proposition \ref{prop:direct-sum-FS}). When $g$ is large, this task becomes computationally difficult, so it useful to formulate the above membership question as a spectrahedral inclusion problem which can benefit from tractable relaxations. To do so, we need to consider the dual object to the matrix jewel (base), which we introduce next. 

\begin{defi}
Consider the vectors $x_1^{(k)}, \ldots, x_k^{(k)} \in \mathbb C^{k-1}$ from Lemma \ref{lem:extremal-points-jewel-base} and define the vectors $y^{(k)}_1, \ldots, y^{(k)}_{k-1}\in \mathbb C^k$ by $y_j(i) = x_i(j)$, for all $i \in [k]$ and $j \in [k-1]$:
	$$y_j^{(k)} = \frac{k}{2}(e_k-e_j), \qquad j \in [k-1].$$
$\mathcal D_{\cuboid,k}$ defined by
\begin{equation*}
\mathcal D_{\cuboid,k}(n):= \Set{X \in (\mathcal M_n^{sa})^{k-1}: \sum_{j = 1}^{k-1} y_j^{(k)} \otimes X_j \leq I_{kn}}. 
\end{equation*}
is called the \emph{matrix cuboid base}. For a $g$-tuple of positive integers $\mathbf k = (k_1, \ldots, k_g)$, we define the \emph{matrix cuboid} $\mathcal D_{\cuboid, \mathbf k}$ to be the free spectrahedron
	$$\mathcal D_{\cuboid, \mathbf k} := \mathcal D_{\cuboid, k_1} \hat \times \mathcal D_{\cuboid, k_2} \hat \times \cdots \hat \times \mathcal D_{\cuboid, k_g},$$
where the Cartesian product operation $\hat \times$ for free spectrahedra was introduced in Equation \eqref{eq:def-Cartesian-product-free-spectrahedra}.
\end{defi}

The definition above generalizes the notion of incompatibility witness from Definition \ref{def:incompatibility-witness-binary} to the setting of $g$ POVMs with arbitrary number of outcomes: $\mathcal D_{\square, g} = \mathcal D_{\cuboid,2^{\times g}}$. Note also that, at level $n=1$, the matrix jewel base and the matrix cuboid base are dual sets; in particular, $\mathcal D_{\cuboid, k}(1)$ is a simplex. 

\begin{remark}
The matrix cuboid is the maximal matrix convex set (in the sense of \cite[Section 4]{davidson2016dilations}, see also Equation \eqref{eq:Wmax}) built on top of the Cartesian product of simplices
$$\mathcal D_{\cuboid, k_1}(1)  \times \mathcal D_{\cuboid, k_2}(1) \times \cdots  \times \mathcal D_{\cuboid, k_g}(1).$$
\end{remark}

We display in Figure \ref{fig:matrix-cuboid} some examples of the $n=1$ of matrix cuboids. 

\begin{figure}[htbp]
	\begin{center}
		\includegraphics[scale=.53]{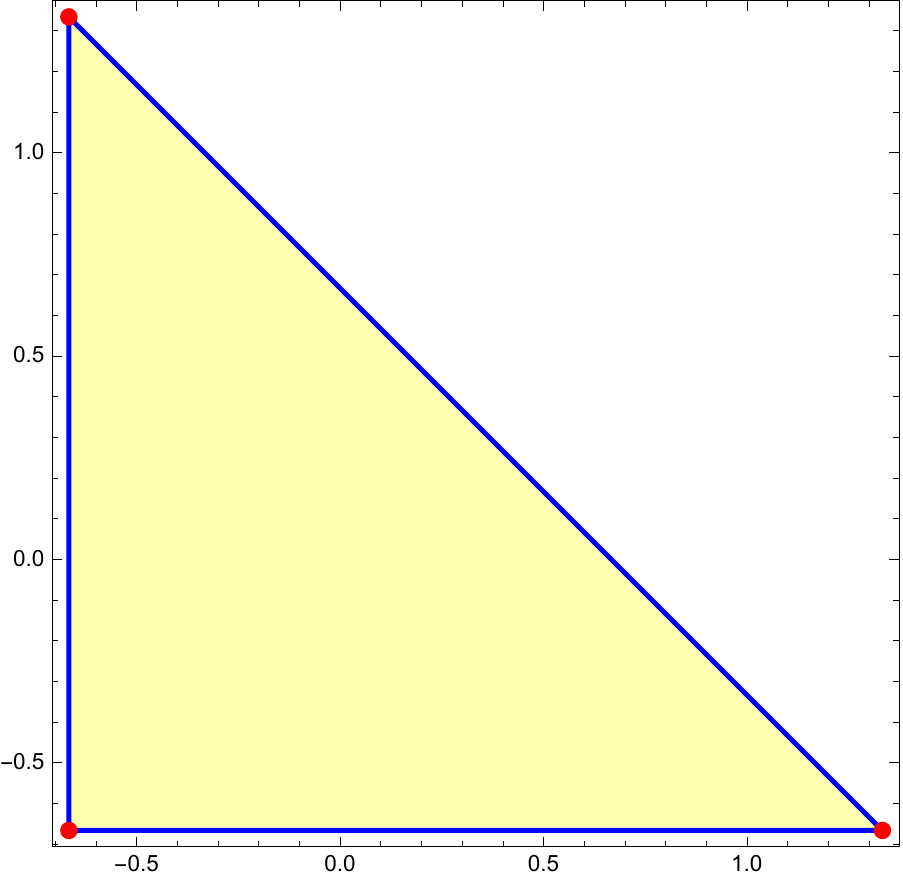}\qquad\qquad\qquad  \includegraphics[scale=.53]{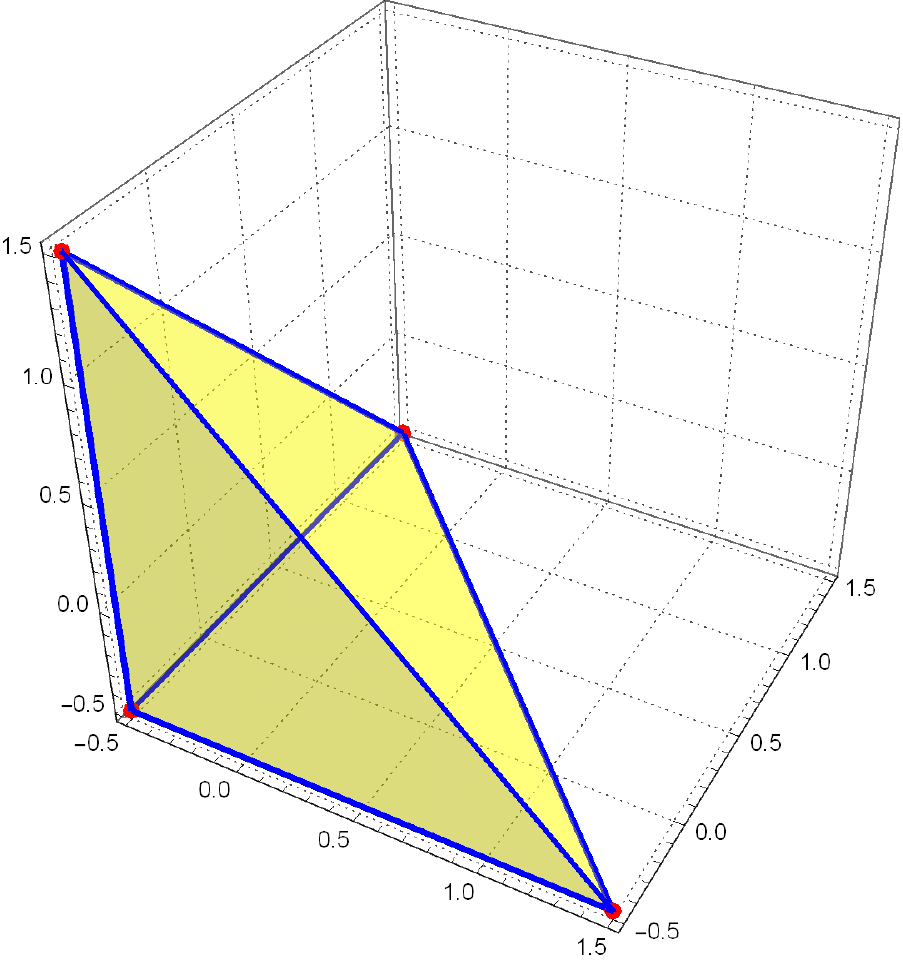}\\
		\includegraphics[scale=.53]{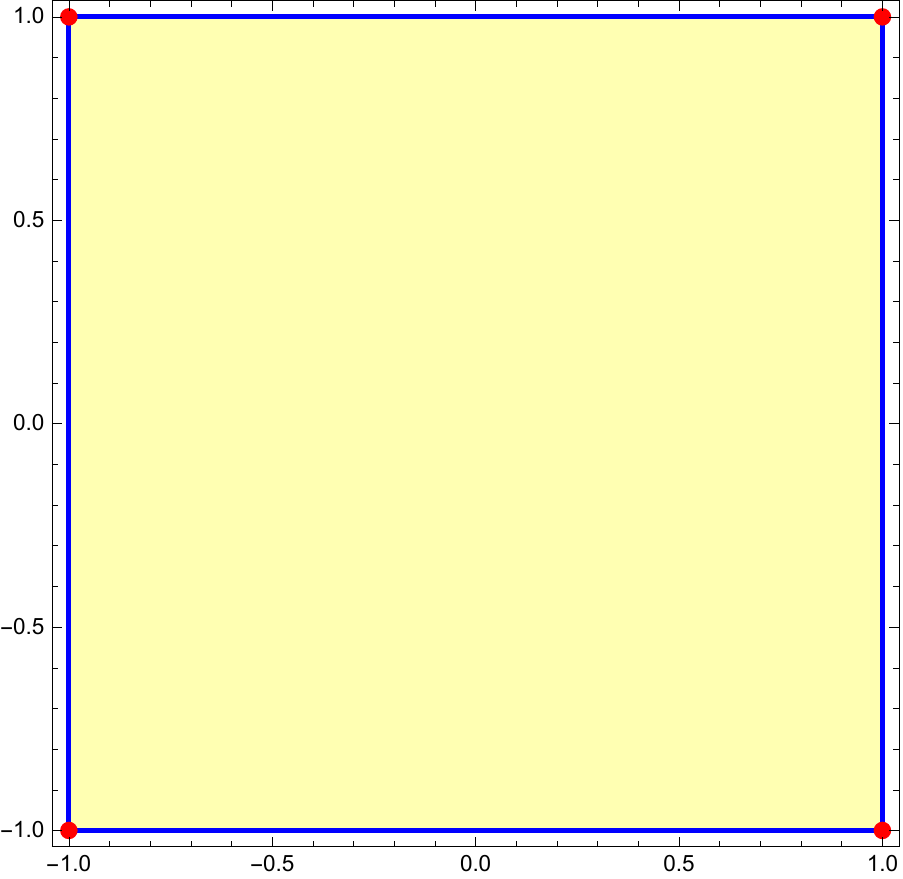}\qquad  \includegraphics[scale=.53]{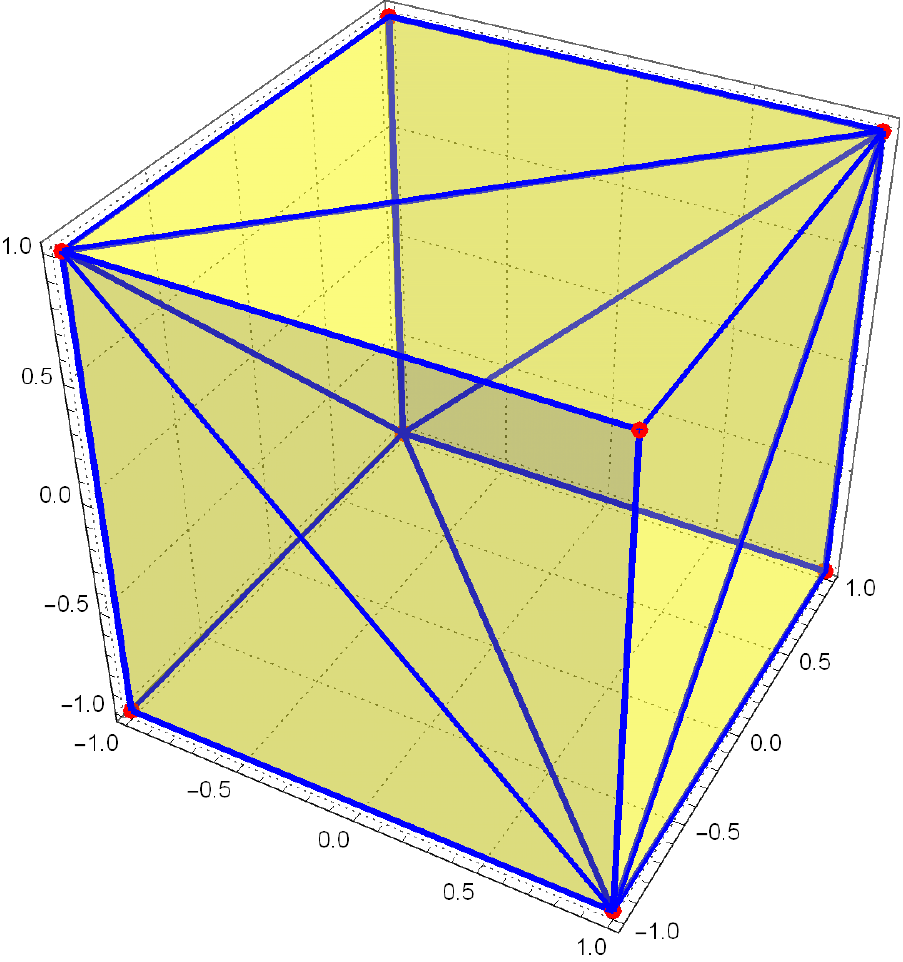}\qquad  \includegraphics[scale=.53]{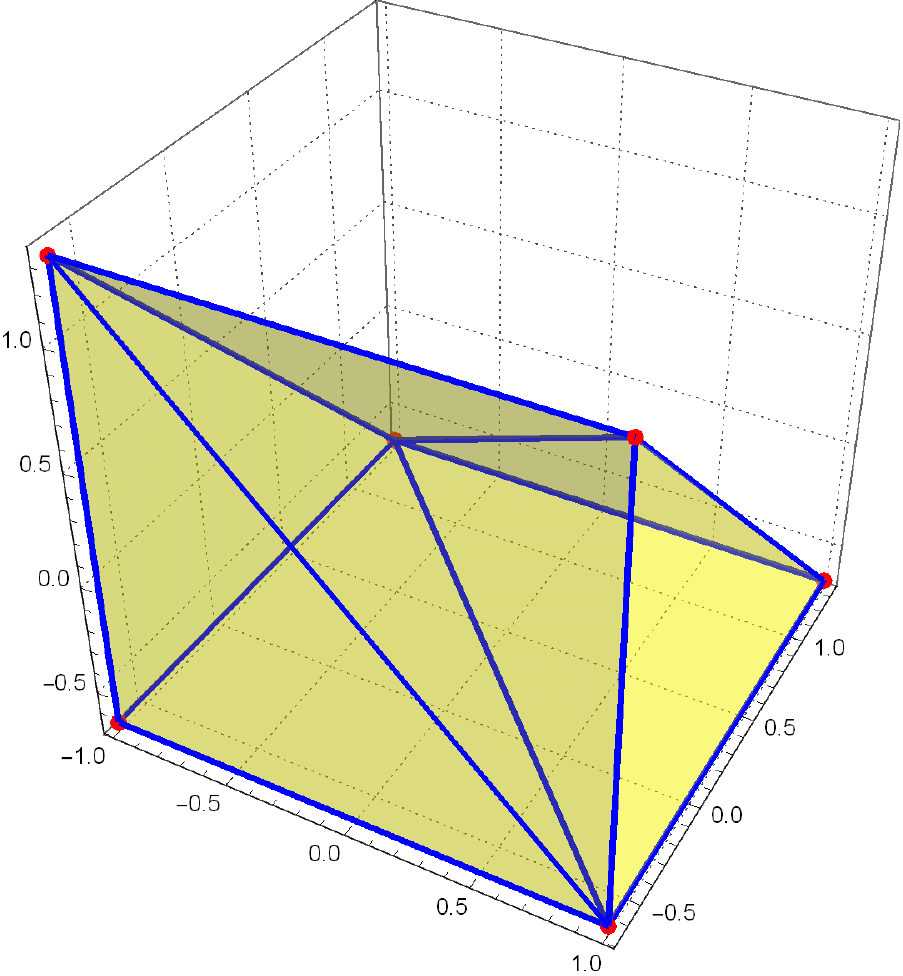}		\caption{Top row: the spectrahedron level of the matrix cuboid base $\mathcal D_{\cuboid,k}(1)$, for $k=3,4$. Bottom row: the spectrahedron level of the matrix cuboids $\mathcal D_{\cuboid,(2,2)}(1)$, $\mathcal D_{\cuboid,(2,2,2)}(1)$, and $\mathcal D_{\cuboid,(2,3)}(1)$. The first two are in fact the matrix cubes $\mathcal D_{\square, 2}(1)$ and $\mathcal D_{\square, 3}(1)$ from \cite{helton2019dilations} (a square and a cube), while the last polyhedron (a triangular prism, the Cartesian product of the triangle and the square) is new.}
		\label{fig:matrix-cuboid}
	\end{center}
\end{figure}

The relation between the notion of incompatibility witness and the matrix cuboid is given in the following result, which generalizes Proposition \ref{prop:incompatibility-witness-vs-cube}. 

\begin{prop}
A $g$-tuple $X \in (\mathcal M_d^{sa})^{\sum_{i=1}^g (k_i-1)}$ is an incompatibility witness if and only if $\mathcal D_{\cuboid, \mathbf k}(1) \subseteq \mathcal D_X(1)$. 
\end{prop}
\begin{proof}
The condition in the statement can be checked at the level of the extreme points of $\mathcal D_{\cuboid, \mathbf k}(1)$, which are Cartesian products of the extreme points
$$\operatorname{ext}\left( \mathcal D_{\cuboid, k_i}(1) \right) = \{w_1, \ldots, w_k\},$$
where $w_i^{(k)} \in \mathbb C^{k-1}$ are given by
$$w_i^{(k)}(j) := - \frac 2 k + 2 \delta_{i,j}, \qquad \forall i \in [k], \, \forall j \in [k-1].$$
Note that the vectors $w_i^{(k)}$ introduced above and the vectors $v_j^{(k)}$ from Definition \ref{def:jewel} are related by $w_i^{(k)}(j)=v_j^{(k)}(i)$, for all $i \in [k]$ and $j \in [k-1]$. From the definition of the Cartesian product, it follows that the extremal points of the matrix jewel base are 
$$\operatorname{ext}\left( \mathcal D_{\cuboid, \mathbf k}(1) \right) = \{w_{\mathbf i}\}_{\mathbf i \in [\mathbf k]},$$
where 
$$w_{\mathbf i}(s,j) = w^{(k_s)}_{i_s}(j), \qquad \forall j \in [k_s], \forall s \in [g].$$
The condition in the statement reads
$$\left(\forall \mathbf i \in [\mathbf k], \qquad \sum_{s=1}^g \sum_{j=1}^{k_s-1} w_{\mathbf i}(s,j) X_{s,j} \leq I \right)\iff   \sum_{s=1}^g \sum_{j=1}^{k_s-1} v_{s,j} \otimes X_{s,j} \leq I ,$$
where $v_{s,j}(\mathbf i) := w_{\mathbf i}(s,j) =  w_{i_s}^{(k_s)}(j) = v_j^{(k_s)}(i_s)$. Equivalently, we have 
$$v_{s,j} = \underbrace{I \otimes \cdots \otimes I}_{s-1 \text{ times}} \otimes v_j^{(k_s)} \otimes\underbrace{I \otimes \cdots \otimes I}_{k-s \text{ times}},$$
which are precisely the vectors defining the matrix jewel, see Equation \eqref{eq:jewel-explicit}.
\end{proof}

\begin{remark}
The above proof also shows that for $k \in \mathbb N$, $\mathcal D_{\jewel, k}(1)^\circ = \mathcal D_{\cuboid, k}(1)$. However, the matrix cuboid and the matrix dual are not dual to each other as matrix convex sets as discussed in Remark \ref{rem:different-direct-sums}.
\end{remark}

\section{Discussion}\label{sec:discussion}

In this section, we study the shape of the inclusion sets for the matrix jewel, before we conclude with some open questions. Contrary to the matrix diamond appearing in the study of binary measurements \cite{bluhm2018joint}, the matrix jewel has not been studied in the literature on free spectrahedra. In algebraic convexity, the matrix convex sets having received the most attention are the matrix cube \cite{ben-tal2002tractable,helton2019dilations}, the different matricial notions of sphere \cite{helton2019dilations,davidson2016dilations}, and the maximal spectrahedra built upon $\ell_p$ spaces \cite{passer2018minimal}. These examples have symmetries that the matrix jewel lacks, rendering its structure more involved. Therefore, we only have two kind of tools at our disposal at this moment to study the structure of the matrix jewel. The first class are the results from quantum information theory presented in Section \ref{sec:compatibility-from-QIT}. The second class of results, derived in Section \ref{sec:symmetrization}, compares the matrix jewel to more symmetric free spectrahedra.

In terms of lower bounds, we have shown in Proposition \ref{prop:cloning-lower} that $\Gamma^{clone}(g,k_{max}d) \subseteq \Delta(g,d,\mathbf k)$, where $k_{max}$ is the maximal entry of $\mathbf k$. This implies in particular that for the balanced case in which $s_1 = \ldots = s_g$, we have
\begin{equation}\label{eq:LB-cloning}
s_{max} \geq \frac{g+k_{max}d}{g(1+k_{max}d)},
\end{equation}
where $s_{max}$ is the greatest balanced inclusion constant in $\Delta(g,d,\mathbf k)$. We also obtain lower bounds from the symmetrization of the matrix jewel (see Theorem \ref{thm:symmetric-jewel})
\begin{equation} \label{eq:symmetrization_again}
\left(\frac{1}{2d(k_1-1)}, \ldots, \frac{1}{2d(k_g-1)}\right)\in \Delta(g,d,\mathbf{k})
\end{equation}
and from the comparison with the matrix diamond (see Theorem \ref{thm:jewel-vs-diamond}) 
\begin{align}
\Delta(g,d,\mathbf{k}) &\supseteq \left(\frac{1}{(k_1-1)^2}, \ldots, \frac{1}{(k_g-1)^2}\right)\cdot \Delta\left(g,d,2^{\times \sum_{i = 1}^g (k_i - 1)}\right) \nonumber \\ &\supseteq \left(\frac{1}{(k_1-1)^2}, \ldots, \frac{1}{(k_g-1)^2}\right)\cdot\mathrm{QC} _{\sum_{i = 1}^g (k_i - 1)}. \label{eq:LB-jewel-vs-diamond}
\end{align}

Let $g_d$ be the maximal number of MUBs which exist in a given dimension $d$. Then, the results gathered in Section \ref{sec:MUB-bounds} translate into upper bounds on $\Delta(g,d,d^{\times g})$, where $g \leq g_d$. For the balanced case, we know from \cite{Designolle2018} that
\begin{equation*}
s_{max} \leq \frac{g + \sqrt{d}}{g(1+\sqrt{d})}.
\end{equation*}
For the asymmetric case, we have from \cite{zhu2015information} that
\begin{equation*}
\Delta(g,d,d^{\times g}) \subseteq \mathrm{QC}_g, \qquad g \leq g_d.
\end{equation*}
Here,
\begin{equation*}
\mathrm{QC}_g := \Set{s \in [0,1]^g: \sum_{i = 1}^g s_i^2 \leq 1}
\end{equation*}
is the higher dimensional equivalent of the positive quarter of the unit circle in two dimensions. For $g = 2$, we have a tighter upper bound, namely the one from \cite{Carmeli2012} (see Proposition \ref{prop:upper-2-mubs}). Let
\begin{equation*}
A = \Set{s \in [0,1]^2: s_1 +s_2 \leq 1} \cup \Set{s \in [0,1]^2: s_1^2 +s_2^2 + \frac{2(d-2)}{2}(1-s_1)(1-s_2) \leq 1}.
\end{equation*}
Then $A \subseteq \mathrm{QC}_2$ with equality for $d=2$ and strict inclusion for $d>2$ and 
\begin{equation*}
\Delta(2,d,d^{\times 2}) \subseteq A.
\end{equation*}
For more general bounds, we can use Proposition \ref{prop:different-k-gamma-inclusion} together with Theorem \ref{thm:delta-is-gamma}. Let $\mathbf{k}$ such that $k_i \geq 2$ for all $i \in [g]$. Then,
\begin{equation*}
\Delta(g,d,\mathbf{k}) \subseteq \Delta(g,d,2^{\times g}).
\end{equation*}
The right hand side was studied in \cite{bluhm2018joint}. From \cite[Theorem VIII.8]{bluhm2018joint}, which uses results from \cite{passer2018minimal}, we obtain
\begin{equation*}
\Delta(g,d,\mathbf{k}) \subseteq \mathrm{QC}_g \qquad \forall d \geq 2^{\left\lceil\frac{g-1}{2}\right\rceil}.
\end{equation*}
Using the concept of inclusion witness, we can bound $\Delta(g,2,2^{\times g})$ for any $g$. We have seen in Corollary \ref{cor:planar-qubits} that
\begin{equation*}
\Delta(g,2,2^{\times g}) \subseteq \Set{s \in [0,1]^g: \sum_{i = 1}^g s_i \leq \frac{1}{\sin(\pi/(2g))}}.
\end{equation*}

\def\arraystretch{2}
\begin{table}[]
	\begin{tabular}{|r|c|l|}
		\hline
		\multicolumn{3}{|c|}{\cellcolor[HTML]{C0C0C0}\textbf{Lower Bounds}} \\ \hline
		cloning                     & \multicolumn{2}{l|}{$s \geq \frac{g+kd}{g(1+kd)}$}       \\ \hline
		symmetrization              & \multicolumn{2}{l|}{$s \geq \frac{1}{2d(k-1)}$}       \\ \hline
		matrix diamond                 & \multicolumn{2}{l|}{ $s \geq \frac{1}{(k-1)^2\sqrt{g(k-1)}}$}       \\ \hline \hline
		\multicolumn{3}{|c|}{\cellcolor[HTML]{C0C0C0}\textbf{Upper Bounds}} \\ \hline
		anti-commuting unitaries    &                       & if $d \geq 2^{\left\lceil\frac{g-1}{2}\right\rceil}$  \\ \cline{1-1} \cline{3-3} 
		MUBs                        & \multirow{-2}{*}{$s\leq \frac{1}{\sqrt{g}}$}   & if $g \leq \text{max nb.~of MUBs in }\mathbb C^d$ \text{ and $k=d$} \\  \hline
	\end{tabular}
\vspace{0.2cm}
\caption{A comparison of all lower and upper bounds for the maximal $s$ such that $(s, \ldots, s) \in \Delta(g,d,\mathbf{k})$, in the case where $\mathbf{k}= (k, \ldots, k)$.}
\label{tbl:compare-all-bounds}
\end{table}

We gather all these bounds in Table \ref{tbl:compare-all-bounds}. In the case where POVMs have the same number of outcomes $k \geq 3$, it turns out that the bound \eqref{eq:symmetrization_again} obtained by symmetrization is always weaker than the cloning bound \eqref{eq:LB-cloning}. Note, however, that this is no longer the case for $\mathbf{k}$ in which not all entries are the same. We compare the cloning bound with the bound \eqref{eq:LB-jewel-vs-diamond} coming from the comparison with the matrix diamond in Figure \ref{fig:LB-compare}. It turns out that in the case where $k=d$ (the number of outcomes matches the dimension), the cloning bound always outperforms the diamond bound, except for qubits ($d=2$).

\begin{figure}
	\centering
	\includegraphics[width=0.4\linewidth]{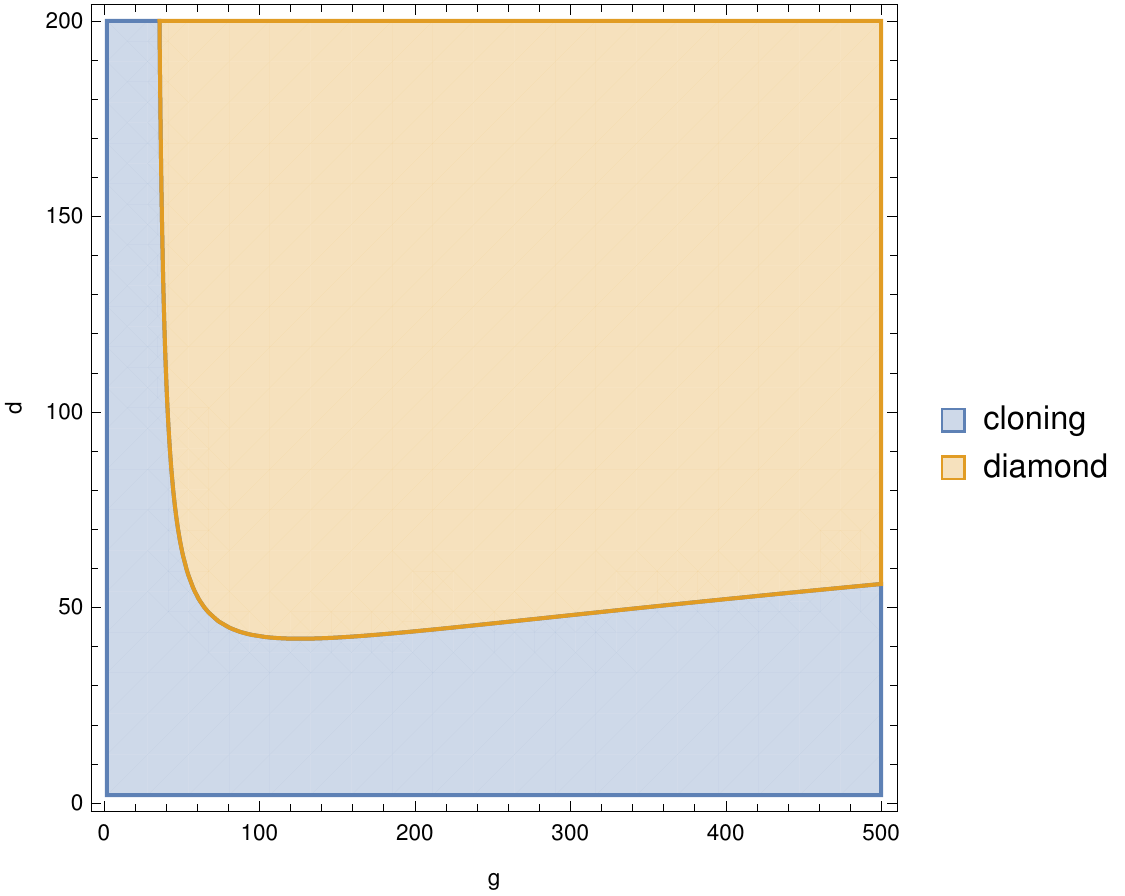} \qquad \qquad 	\includegraphics[width=0.4\linewidth]{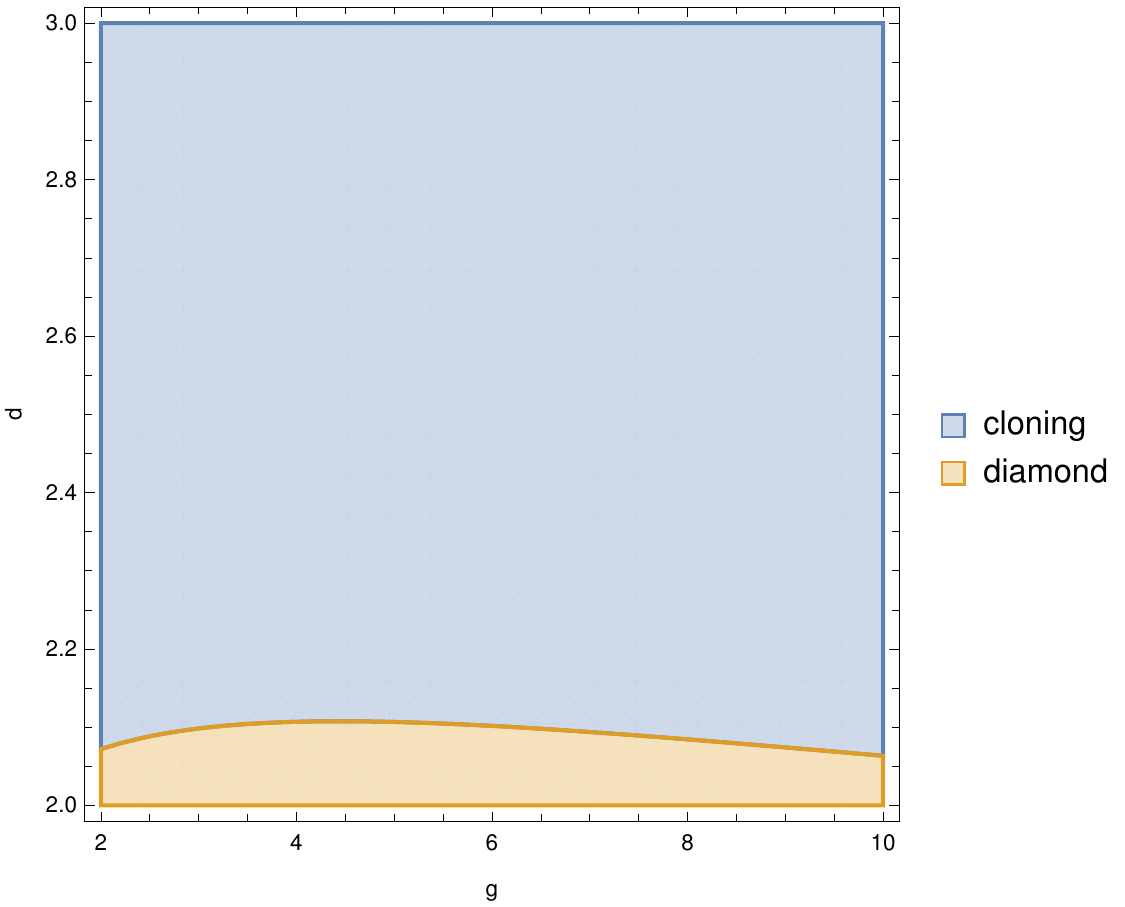}
	\caption{A comparison of the two lower bounds from equations \eqref{eq:LB-cloning} and \eqref{eq:LB-jewel-vs-diamond}, coming respectively from quantum cloning and from the comparison to the matrix diamond. On the left panel, we consider the case of $g$ POVMs on $\mathbb C^d$ with $k=3$ outcomes, while on the right panel we consider the case $k=d$. The regions correspond to the better (i.e.~larger) lower bound.}
	\label{fig:LB-compare}
\end{figure}

For the balanced compatibility region $\Gamma(g,d, \mathbf{k})$, the lower bounds obtained via the symmetrization of the matrix jewel in Theorems \ref{thm:symmetric-jewel} and \ref{thm:jewel-vs-diamond} are new and improve over the lower bounds from asymmetric cloning for suitable choices of parameters (see Figure \ref{fig:LB-compare}). The correspondence in Theorem \ref{thm:delta-is-gamma} yields
\begin{equation*}
\left(\frac{1}{2d(k_1-1)}, \ldots, \frac{1}{2d(k_g-1)}\right)\in \Gamma(g,d,\mathbf{k})
\end{equation*}
and
\begin{equation*}
\Gamma(g,d,\mathbf{k}) \supseteq \left(\frac{1}{(k_1-1)^2}, \ldots, \frac{1}{(k_g-1)^2}\right)\cdot\mathrm{QC} _{\sum_{i = 1}^g (k_i - 1)}. 
\end{equation*}

As mentioned at the beginning of this section, all the bounds on the inclusion set for the matrix jewel we have obtained here stem either from quantum information theory or from some symmetrization technique. We leave it as an open question whether it is possible to obtain stronger bounds from the study of free spectrahedra, which would then have interesting consequences for quantum information theory. We also leave open the study of the matrix cuboid from Section \ref{sec:incompatibility-witnesses-general}, which can be seen as a generalization of the matrix cube. In particular, the inclusion constants for such free spectrahedra would allow to obtain, via Proposition \ref{prop:incompatibility-witness-vs-cube}, efficient criteria for deciding whether a tuple of matrices is an incompatibility witness for general POVMs.

\bigskip

\noindent {\it Acknowledgments.} A.B. acknowledges support from the ISAM Graduate Center at the Technische Universit\"at M\"unchen and financial support from the VILLUM FONDE Nvia the QMATH Centre of Excellence (Grant no. 10059). Furthermore, A.B. acknowledges support from the QuantERA ERA-NET Cofund in Quantum Technologies implemented within the European Union's Horizon 2020 Programme (QuantAlgo project) via the Innovation Fund Denmark.  I.N. and A.B. would like to thank Guillaume Aubrun for pointing us to the direct sum of convex sets. I.N.'s research has been supported by the ANR projects {StoQ} (grant number ANR-14-CE25-0003-01) and {NEXT} (grant number ANR-10-LABX-0037-NEXT), by the PHC Sakura program (grant number 38615VA), and by the UEFISCDI (grant number PN-III-P1-1.1-MCT-2018-0015). Both authors also acknowledge the hospitality of M.~Jivulescu and N.~Lupa from the Universitatea Politehnic\u{a} Timi\c{s}oara, where most of this work was done. Finally, both authors would like to thank the anonymous referees for their constructive comments. The overall quality of the presentation has improved significantly thanks to their reports.

\bibliographystyle{alpha}
\bibliography{../spectralit}

\newcommand{\etalchar}[1]{$^{#1}$}
\begin{thebibliography}{DDOSS17}

\bibitem[Arv72]{Arveson1972}
William Arveson.
\newblock {Subalgebras of C$^\ast$-algebras II}.
\newblock {\em Acta Mathematica}, 128:271--308, 1972.

\bibitem[Bar02]{Barvinok2002}
Alexander Barvinok.
\newblock {\em A course in convexity}, volume~54 of {\em Graduate Studies in
  Mathematics}.
\newblock American Mathematical Society, 2002.

\bibitem[BCP{\etalchar{+}}14]{Brunner2014}
Nicolas {Brunner}, Daniel {Cavalcanti}, Stefano {Pironio}, Valerio {Scarani},
  and Stephanie {Wehner}.
\newblock {Bell nonlocality}.
\newblock {\em Reviews of Modern Physics}, 86:419--478, 2014.

\bibitem[BN18]{bluhm2018joint}
Andreas Bluhm and Ion Nechita.
\newblock Joint measurability of quantum effects and the matrix diamond.
\newblock {\em Journal of Mathematical Physics}, 59(11):112202, 2018.

\bibitem[Boh28]{Bohr1928}
Niels Bohr.
\newblock The quantum postulate and the recent development of atomic theory.
\newblock {\em Nature}, 121(3050):580--590, 1928.

\bibitem[Bre97]{Bremner1997}
David~D. Bremner.
\newblock On the complexity of vertex and facet enumeration for complex
  polytopes.
\newblock {\em Ph.D. thesis, School of Computer Science, McGill University,
  Monr{\'e}al, Canada}, 1997.

\bibitem[BTN02]{ben-tal2002tractable}
Aharon Ben-Tal and Arkadi Nemirovski.
\newblock On tractable approximations of uncertain linear matrix inequalities
  affected by interval uncertainty.
\newblock {\em SIAM Journal on Optimization}, 12(3):811--833, 2002.

\bibitem[CHT12]{Carmeli2012}
Claudio Carmeli, Teiko Heinosaari, and Alessandro Toigo.
\newblock Informationally complete joint measurements on finite quantum
  systems.
\newblock {\em Physical Review A}, 85:012109, Jan 2012.

\bibitem[CHT18]{carmeli2018quantum}
Claudio Carmeli, Teiko Heinosaari, and Alessandro Toigo.
\newblock Quantum incompatibility witnesses.
\newblock {\em arXiv preprint arXiv:1812.02985}, 2018.

\bibitem[DDOSS17]{davidson2016dilations}
Kenneth~R. Davidson, Adam Dor-On, Orr~Moshe Shalit, and Baruch Solel.
\newblock Dilations, inclusions of matrix convex sets, and completely positive
  maps.
\newblock {\em International Mathematics Research Notices},
  2017(13):4069--4130, 2017.

\bibitem[DEB{\.Z}10]{durt2010mutually}
Thomas Durt, Berthold-Georg Englert, Ingemar Bengtsson, and Karol
  {\.Z}yczkowski.
\newblock On mutually unbiased bases.
\newblock {\em International Journal of Quantum Information}, 8(04):535--640,
  2010.

\bibitem[DSFB18]{Designolle2018}
S\'ebastien Designolle, Paul Skrzypczyk, Florian Fr\"owis, and Nicolas Brunner.
\newblock Quantifying measurement incompatibility of mutually unbiased bases.
\newblock {\em arXiv preprint arXiv:1805.09609}, 2018.

\bibitem[EW97]{Effros1997}
Edward~G. Effros and Soren Winkler.
\newblock Matrix convexity: Operator analogues of the bipolar and {Hahn–Banach}
  theorems.
\newblock {\em Journal of Functional Analysis}, 144(1):117 -- 152, 1997.

\bibitem[Fin82]{Fine1982}
Arthur Fine.
\newblock {Hidden variables, joint probability, and the Bell inequalities}.
\newblock {\em Physical Review Letters}, 48(5):291--295, 1982.

\bibitem[Has17]{hashagen2016universal}
Anna-Lena Hashagen.
\newblock Universal asymmetric quantum cloning revisited.
\newblock {\em Quantum Information \& Computation}, 17(9-10):0747--0778, 2017.

\bibitem[Hei27]{Heisenberg1927}
Werner Heisenberg.
\newblock {{\"U}ber den anschaulichen Inhalt der quantentheoretischen Kinematik
  und Mechanik}.
\newblock {\em Zeitschrift f{\"u}r Physik}, 43(3):172--198, 1927.

\bibitem[HKM13]{helton_matricial_2013}
J.~William Helton, Igor Klep, and Scott McCullough.
\newblock The matricial relaxation of a linear matrix inequality.
\newblock {\em Mathematical Programming}, 138(1-2):401--445, 2013.

\bibitem[HKMS19]{helton2019dilations}
J.~William Helton, Igor Klep, Scott McCullough, and Markus Schweighofer.
\newblock Dilations, linear matrix inequalities, the matrix cube problem and
  beta distributions.
\newblock {\em Memoirs of the American Mathematical Society}, 257(1232), 2019.

\bibitem[HKR15]{Heinosaari2015}
Teiko Heinosaari, Jukka Kiukas, and Daniel Reitzner.
\newblock Noise robustness of the incompatibility of quantum measurements.
\newblock {\em Physical Review A}, 92:022115, 2015.

\bibitem[HMZ16]{Heinosaari2016}
Teiko Heinosaari, Takayuki Miyadera, and M\'ario Ziman.
\newblock An invitation to quantum incompatibility.
\newblock {\em Journal of Physics A: Mathematical and Theoretical},
  49(12):123001, 2016.

\bibitem[HZ11]{Heinosaari2011}
Teiko Heinosaari and M{\'a}rio Ziman.
\newblock {\em The Mathematical Language of Quantum Theory}.
\newblock Cambridge University Press, 2011.

\bibitem[Jen18]{jencova2018incompatible}
Anna Jen{\v{c}}ov{\'a}.
\newblock Incompatible measurements in a class of general probabilistic
  theories.
\newblock {\em Physical Review A}, 98(1):012133, 2018.

\bibitem[Kay16]{kay2016optimal}
Alastair Kay.
\newblock Optimal universal quantum cloning: Asymmetries and fidelity measures.
\newblock {\em Quantum Information \& Computation}, 16(11 \& 12):0991--1028,
  2016.

\bibitem[Key02]{Keyl2002}
Michael Keyl.
\newblock Fundamentals of quantum information theory.
\newblock {\em Physics Reports}, 369(5):431--548, 2002.

\bibitem[KHF14]{kunjwal2014quantum}
Ravi Kunjwal, Chris Heunen, and Tobias Fritz.
\newblock Quantum realization of arbitrary joint measurability structures.
\newblock {\em Physical Review A}, 89(5):052126, 2014.

\bibitem[Pas18]{Passer2018a}
Benjamin Passer.
\newblock Shape, scale, and minimality of matrix ranges.
\newblock {\em arXiv preprint arXiv:1803:09212}, 2018.

\bibitem[Pau03]{Paulsen2002}
Vern Paulsen.
\newblock {\em Completely Bounded Maps and Operator Algebras}, volume~78 of
  {\em Cambridge Studies in Advanced Mathematics}.
\newblock Cambridge University Press, 2003.

\bibitem[PP19]{Passer2019}
Benjamin Passer and Vern~I. Paulsen.
\newblock Matrix range characterizations of operator system properties.
\newblock {\em arXiv preprint arXiv:1912.06279}, 2019.

\bibitem[PSS18]{passer2018minimal}
Benjamin Passer, Orr~Moshe Shalit, and Baruch Solel.
\newblock Minimal and maximal matrix convex sets.
\newblock {\em Journal of Functional Analysis}, 274:3197--3253, 2018.

\bibitem[RS80]{Reed1980}
Michael Reed and Barry Simon.
\newblock {\em Methods of modern mathematical physics: Functional Analysis}.
\newblock Academic Press, revised and enlarged edition edition, 1980.

\bibitem[S{\'C}HM14]{studzinski2014group}
Micha{\l} Studzi{\'n}ski, Piotr {\'C}wikli{\'n}ski, Micha{\l} Horodecki, and
  Marek Mozrzymas.
\newblock {Group-representation approach to $1 \to N$ universal quantum cloning
  machines}.
\newblock {\em Physical Review A}, 89(5):052322, 2014.

\bibitem[ULMH16]{uola2016adaptive}
Roope Uola, Kimmo Luoma, Tobias Moroder, and Teiko Heinosaari.
\newblock Adaptive strategy for joint measurements.
\newblock {\em Physical Review A}, 94(2):022109, 2016.

\bibitem[UMG14]{Uola2014}
Roope Uola, Tobias Moroder, and Otfried G\"uhne.
\newblock Joint measurability of generalized measurements implies classicality.
\newblock {\em Physical Review Letters}, 113:160403, 2014.

\bibitem[Wat18]{watrous2018theory}
John Watrous.
\newblock {\em The Theory of Quantum Information}.
\newblock Cambridge University Press, 2018.

\bibitem[Wer98]{Werner1998}
Reinhard~F. Werner.
\newblock Optimal cloning of pure states.
\newblock {\em Physical Review A}, 58(3):1827--1832, 1998.

\bibitem[WF89]{wootters1989optimal}
William~K. Wootters and Brian~D. Fields.
\newblock Optimal state-determination by mutually unbiased measurements.
\newblock {\em Annals of Physics}, 191(2):363--381, 1989.

\bibitem[Wit84]{Wittstock1984}
Gerd Wittstock.
\newblock On matrix order and convexity.
\newblock In Klaus-Dieter Bierstedt and Benno Fuchssteiner, editors, {\em
  Functional Analysis: Surveys and Recent Results III}, volume~90 of {\em
  North-Holland Mathematics Studies}, pages 175 -- 188. North-Holland, 1984.

\bibitem[ZHC16]{zhu2016universal}
Huangjun Zhu, Masahito Hayashi, and Lin Chen.
\newblock Universal steering criteria.
\newblock {\em Physical Review Letters}, 116(7):070403, 2016.

\bibitem[Zhu15]{zhu2015information}
Huangjun Zhu.
\newblock Information complementarity: A new paradigm for decoding quantum
  incompatibility.
\newblock {\em Scientific reports}, 5:14317, 2015.

\end{thebibliography}
\end{document}